\documentclass[aps,pra,reprint,showkeys]{revtex4-2}
\usepackage{amsmath,paralist,amsthm,comment,amssymb}
\usepackage{times,color}
\usepackage{epsfig}
\usepackage{amscd}
\usepackage{tikz-cd}
\usepackage{graphicx,float}
\graphicspath{ {./Figures/} }
\usepackage{caption}
\usepackage{subcaption}
\usepackage{empheq}
\usepackage{appendix}
\usepackage{tikz}
\usetikzlibrary{matrix,arrows}
\usetikzlibrary{shapes.geometric}
\usetikzlibrary{patterns} 

\usepackage{pgfplots}

\usepackage{comment}

\usepackage[breaklinks=true,colorlinks=true,linkcolor=blue,urlcolor=blue,citecolor=red]{hyperref}

\usepackage{float}
\makeatletter
\let\newfloat\newfloat@ltx
\makeatother
\usepackage{algorithm}
\usepackage{algpseudocode}
\algrenewcommand\algorithmicrequire{\textbf{Input:}}
\algrenewcommand\algorithmicensure{\textbf{Output:}}

\usepackage{datetime}

\long\def\/*#1*/{}

\newcommand{\nix}[1]{}

\newtheorem{theorem}{Theorem}

\newtheorem{lemma}[theorem]{Lemma}

\newtheorem{remark}[theorem]{Remark}

\newtheorem{definition}[theorem]{Definition}

\begin{document}
\title{Quantum computation with charge-and-color-permuting twists in qudit color codes}
\author{Manoj G. Gowda}
\author{Pradeep Kiran Sarvepalli}
\affiliation{
Department of Electrical Engineering, Indian Institute of Technology Madras, Chennai 600 036, India
}


\begin{abstract}
Twists are defects in the lattice which can be utilized to perform computations on encoded data.
Twists have been studied in various classes of topological codes like qubit and qudit surface codes, qubit color codes and qubit subsystem color codes.
They are known to exhibit projective non-Abelian statistics which is exploited to perform encoded gates.
In this paper, we initiate the study of twists in qudit color codes over odd prime alphabet.
To the best of our knowledge, this is the first study of twists in qudit color codes.
Specifically, we present a systematic construction of twists in qudit color codes that permute both charge and color of the excitations.
We also present a mapping between generalized Pauli operators and strings in the lattice.
Making use of the construction, we give protocols to implement generalized Clifford gates using charge-and-color-permuting twists.
\end{abstract}
\maketitle

\section{Introduction}
Qudits are quantum systems with $d$ levels and can be thought of as generalization of qubits which are two-dimensional quantum systems.
There is significant literature on qudit quantum error correcting codes~\cite{Ashikhmin2001,Grassl2003,Ketkar2006,Sarvepalli2010,You2013,Brell2015,Watson2015, Hutter2015,Marks2017,Haah2021}.
Qudit based quantum computation technologies are gaining pace.
Recently, universal quantum computation with qudits has been realized using optics~\cite{Niu2018,Paesani2021} and trapped ions~\cite{Low2020}.
Using qudits can lead to significant reduction in the number of operations and improvement in circuit depth in the realization of Toffoli gate~\cite{Kiktenko2020}.
Also, many quantum algorithms have been realized using qudits~\cite{Gedik2015, Zhang2019,Lu2019}.
These developments provide some impetus to our study of quantum computation with qudit quantum codes. 
From the perspective of performance, qudit codes exhibit higher threshold to (generalized) depolarizing noise and the threshold increases with increasing dimension of qudits~\cite{Duclos-Cianci2013, Anwar2014, Marks2017, AloshiousSarvepalli2019}.
Moreover, qudit codes pose challenges not encountered in the qubit case and are of interest from a theoretical point of view as well. 

Color codes are an important class of topological quantum codes, they were introduced by Bombin~\textit{et al.}~\cite{Bombin2006}.
They were proposed with a view to enable  transversal quantum gates providing a route to fault-tolerant quantum computation. 
Color codes were generalized to prime  alphabet in Ref.~\cite{Sarvepalli2010}
and shown to support  transversal implementation of the generalized Clifford group. 
Brell further generalized color codes based on finite groups to support non-Abelian excitations~\cite{Brell2015}.
Watson \textit{et al.}~\cite{Watson2015} showed that qudit color codes support generalized Clifford gates in higher spatial dimensions, specifically, the phase gate from higher levels of the Clifford hierarchy.

An alternative route to fault-tolerance through topological codes is to use topological defects like holes and twists in conjunction with code deformation~\cite{Bombin2011, Fowler2011, Landahl2011, Fowler2012,Hastings2015,Brown2017, Yoder2017, Lavasani2018, GowdaSarvepalli2020, GowdaSarvepalli2021}. 
Specifcally, such methods have been studied for color codes in Refs.~\cite{Fowler2011, Landahl2011}.
Twists are a form of defects introduced in the lattice by spoiling some property of the lattice.
Kitaev first suggested the use of twists to encode quantum information~\cite{Kitaev2003}.
Related work on twists in surface codes can be found in Refs.~\cite{Bombin2010, Yu-Wen2012, You2013, Hastings2015,Brown2017, Lavasani2018, GowdaSarvepalli2020}.
Twists in qubit color codes were first studied by Kesselring \textit{et al.}~\cite{Kesselring2018}.
They also developed a framework for the study of twists in color codes using the theory of domain walls and  cataloged all possible types of twists in qubit color codes.
Their work showed that there are a large number of twist defects possible in color codes. 
This is in sharp contrast to the qubit surface code case where we have only one type of twist defects, namely the charge permuting twists.

In Ref.~\cite{GowdaSarvepalli2021}, the authors studied a construction of charge permuting and color permuting twists from an arbitrary $2$-colex.
For the proposed codes, they also gave protocols to implement encoded generalized Clifford gates.
Twists were studied in topological subsystem color codes in Ref.~\cite{Bombin2011}, where it was shown that Clifford gates can be realized by braiding twists.
Litinski and von Oppen~\cite{Litinski2018-2} discussed twists in Majorana fermion code and adapted the technique of twist-based lattice surgery to fermionic codes.

\medskip

\noindent \emph{Contributions.} 
While there have been many prior studies on qudit color codes~\cite{Sarvepalli2010, Hutter2015, Brell2015,Watson2015,Marks2017, AloshiousSarvepalli2019, Haah2021}, 
to the best of our knowledge,
there is no literature on using defects like twists and holes to encode and process quantum information in qudit color codes. 
This motivates our study of twists in qudit color codes and their application to quantum computation.

In this paper, we focus on charge-and-color-permuting twists in qudit color codes over hexagonal lattices but the ideas are applicable to general lattices.
We assume that the qudits are $d$-dimensional where $d$ is an odd prime. 
These twists permute both charge and color of syndromes.
Twists can be introduced in a lattice either by lattice modification (and hence redefining stabilizers on the modified faces in the lattice) or without making modifications to the lattice (charge permuting twists in qubit color codes do not require lattice modification but only stabilizer modification on certain faces~\cite{Kesselring2018, GowdaSarvepalli2021}).
We introduce charge-and-color-permuting twists by lattice modification and a careful assignment of stabilizer generators in the modified lattice.

Some of the encoded generalized Clifford gates are realized by braiding twists.
Twists are braided by modifying lattice and also stabilizers.
During braiding some of the faces (and hence the stabilizers defined on them) in the lattice can grow in size.
This is undesirable as it may lead to high weight stabilizers. 
This poses the challenge of braiding twists in a way that large faces are avoided.
This problem is circumvented by moving twists in such a way that high weight stabilizers are avoided.  

Our contributions are listed below:
\begin{compactenum}[i)]
\item We propose a systematic construction of charge-and-color-permuting twists in qudit color code lattices starting from a $2$-colex.
The lattice modification performed is similar to the one used to create color permuting twists and we define the stabilizers on the modified lattice to obtain charge-and-color-permuting twists in qudit color codes.
The construction is summarized in Theorem~\ref{thm:construction-summary}.
\item We present a mapping between generalized Pauli operators and strings on color codes with and without (charge-and-color-permuting) twists. 
This mapping comes in handy while implementing encoded gates by braiding.
\item  We propose protocols for the implementation of generalized Clifford group.
Broadly, we use the techniques of braiding twists and Pauli frame update.
Pauli frame update is used for realizing multiplier and DFT gates while braiding is used for phase and CNOT gates. 
\end{compactenum}

One of the appealing features of the proposed protocols for gates is that the lattice modifications are kept simple.
Qudits are added or disentangled from  the code lattice while being retained in the underlying lattice.
The code lattice retains much of the regular structure of the underlying lattice which makes tracking the lattice modifications and code deformation simple.

For the qudit color codes in Refs.~\cite{Sarvepalli2010,Watson2015}, many encoded gates can be realized  transversally.
Transversal gates exist in the case of $G$-color codes~\cite{Brell2015} when the group $G$ is Abelian.
The number of physical qudits involved during gate implementation is large in comparison with our proposal. 
However, the circuits for implementing encoded gates have unit depth whereas in our proposal we expect it to be of the order of code distance. 
The qudit color codes in Ref.~\cite{Watson2015} are over over higher spatial dimensions and they can implement a non-Clifford gate transversally. However, the proposed codes are planar and maybe preferable for practical reasons.

\medskip

\noindent \emph{Organization.} This paper is organized as follows.
In Sec.~\ref{sec:background}, we discuss the preliminary material related to  qudit color codes.
We present construction of charge-and-color-permuting twists in Sec.~\ref{sec:charge-color}.
We introduce a mapping of generalized Pauli operators to strings in Sec.~\ref{sec:pauli-string-mapping}.  
We present implementation of gates using twists in Sec.~\ref{sec:gates}.
Some additional details and proofs are relegated to the appendices.

\section{Background}
\label{sec:background}
\subsection{Generalized Pauli operators}
Qudits are quantum systems with $d$ levels.
When $d = 2$, we have qubits.
In this paper, we assume that $d$ is an odd prime.
Let $\mathbb{F}_d$ be a finite field with $d$ elements.
Generalized Pauli operators, also known as Heisenberg-Weyl operators~\cite{Marks2017}, in the qudit case are defined as below~\cite{Ashikhmin2001, Grassl2003, Ketkar2006, Haah2021}:
\begin{subequations}
\begin{eqnarray}
X(a) &=& \sum_{x \in \mathbb{F}_d} | x + a \rangle \langle x|,\\
Z(b) &=& \sum_{x \in \mathbb{F}_d} \omega^{bx} |x \rangle \langle x |,
\end{eqnarray}
\end{subequations}
where $+$ denotes addition modulo $d$ and $\omega = e^{2 \pi i / d}$.
Note that $X(d) = Z(d) = I$ and $X(a) X(b) = X(a + b)$ and $Z(a) Z(b) = Z(a + b)$.
The generalized Pauli group $\mathcal{P}$ is generated by the operators $X(a)$ and $Z(b)$:
\begin{equation}
    \mathcal{P} = \langle \omega^c X(a) Z(b) \vert a,b,c \in \mathbb{F}_d \rangle.
\end{equation}
The operators $X(a)$ and $Z(b)$ obey the commutation relation
\begin{equation}
    Z(b) X(a) = \omega^{ab} X(a) Z(b).
    \label{eq:qudit-XZ-commutation}
\end{equation}
The generalized Pauli operators encountered most often in this paper are $Z(1)$, $Z(d-1)$, $X(1)$ and $X(d-1)$.
They are denoted as $Z$, $Z^\dagger$, $X$ and $X^\dagger$ respectively.

\subsection{Qudit color codes}
To define a qudit color code, we embed a trivalent and three face colorable lattice on a two-dimensional surface.
Such lattices are called 2-colexes. 
As the faces are 3-colorable, we denote by $\mathsf{F}_c$ the set of faces with color $c \in \{r,g,b\}$.
Further 2-colexes can be shown to be bipartite, see for instance~\cite{Sarvepalli2010}. 
In other words, the set of vertices $\mathsf{V}$ in such lattices can be partitioned into two sets, $\mathsf{V}_e$ and $\mathsf{V}_o$ such that no two neighboring vertices will be in the same set.
The set of vertices of a face $f$ is denoted by $V(f)$.
Qudits are placed on the vertices of the lattice and two stabilizer generators are defined on every face:
\begin{subequations}
	\begin{eqnarray}
		B_f^{Z} &=& \prod_{v \in V(f)} Z_v\\ 
		B_f^{X} &=& \prod_{v \in V(f) \cap \mathsf{V}_e} X_v \prod_{v \in V(f) \cap \mathsf{V}_o} X_v^\dagger
	\end{eqnarray}
	\label{eqn:qudit-stab}
\end{subequations}

Note that the $Z$ type stabilizer $B_f^Z$ contains only $Z$ operator and it is independent of the vertex type whereas in  the $X$ type stabilizer $B_f^X$, both $X$ and $X^\dagger$ occur depending on  the type of vertex.

The stabilizers defined in Equations~\eqref{eqn:qudit-stab} are shown in Fig.~\ref{fig:qd-cc-stabilizers}.
Note that any two adjacent faces $f_1$ and $f_2$ share exactly an edge $(u,v)$ where $u \in \mathsf{V}_e$ and $v \in  \mathsf{V}_o$.
The restriction of stabilizers to the common vertices will be $Z_u Z_v$ and $X_u X^\dagger_v$.
Since, these operators commute, the corresponding stabilizers defined on faces $f_1$ and $f_2$ commute.
Therefore, the face stabilizers defined in Equations~\eqref{eqn:qudit-stab} commute. 

\begin{figure}[htb]
    \centering
    \includegraphics[scale = .85]{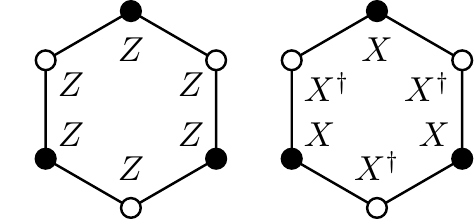}
    \caption{Stabilizers for the qudit color code. Dark and light circles indicate even and odd vertices respectively.} 
    \label{fig:qd-cc-stabilizers}
\end{figure}

The constraints satisfied by these stabilizers is same as that of qubit color codes~\cite{Sarvepalli2010}:
\begin{subequations}
\begin{eqnarray}
\prod_{f \in \mathsf{F}_r} B_f^X &=& \prod_{f \in \mathsf{F}_g} B_f^X = \prod_{f \in \mathsf{F}_b} B_f^X,\\
\prod_{f \in \mathsf{F}_r} B_f^Z &=& \prod_{f \in \mathsf{F}_g} B_f^Z = \prod_{f \in \mathsf{F}_b} B_f^Z.
\end{eqnarray}
\label{eqn:qudit-stab-constraint}
\end{subequations}
There are four constraints and hence four dependent stabilizers.
A graph with $n$ vertices (qudits) embedded on a surface of genus $g$ defines an $[[n, 4g]]_d$ quantum code.

\noindent \emph{Embedding surface and boundaries.} In the rest of this paper, we assume that the graph is embedded on two-dimensional plane.
We also assume that the unbounded face has the same color throughout, as in the case of qubit color codes with twists~\cite{GowdaSarvepalli2021}, see Fig.~\ref{fig:Qd-cc-hexagon-lattice-boundary}.
Multiple edges are introduced along the boundary so that all vertices are trivalent.
It can be shown that such lattices do not encode any logical qudits~\cite{GowdaSarvepalli2021}.

\begin{figure}[htb]
    \centering
    \includegraphics[scale = 1]{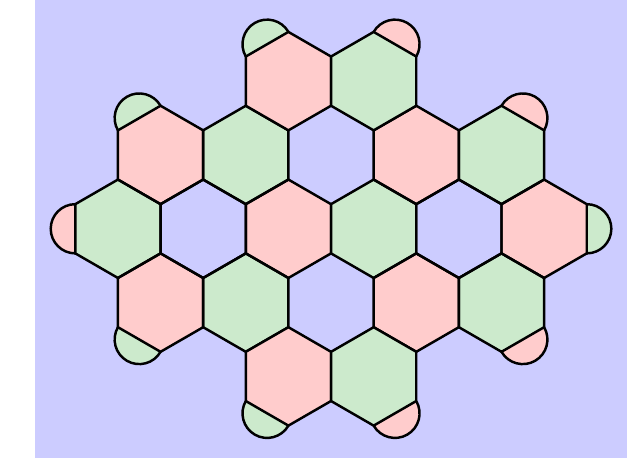}
    \caption{Hexagon lattice with boundary. The unbounded face has blue color and multiple edges are introduced along the boundary to preserve trivalency of vertices.}
    \label{fig:Qd-cc-hexagon-lattice-boundary}
\end{figure}

\subsection{Syndromes  (excitations) in qudit color codes}
Suppose that we have an error $E$, then measuring a stabilizer generator $S$ produces a syndrome depending on the commutation relation between $E$ and $S$. 
Specifically, we have $SE = \omega^a ES$, where $a\in \{ 0, 1, \ldots, d-1\}$.
The eigenvalue $\omega^a$ or more simply $a$ is called the syndrome obtained on measuring $S$.
These syndromes are also called excitations or charges. 
We use these terms interchangeably. 

Suppose that a qudit (in a color code) undergoes $X$ error (recall that we denote the operators $X(1)$ and $Z(1)$ as $X$ and $Z$ respectively), see Fig.~\ref{fig:Qd-cc-X-error-syndrome}.
An $X$ error violates $Z$ type stabilizers on the three faces incident on the vertex (qudit) and hence these three faces will host syndromes.
We can think of the syndromes on the faces as  excitations or quasiparticles which can be labeled by the color of the face
on which they are hosted and the magnitude of the syndrome.
The excitations caused by $X$ type errors are called magnetic charges and denoted $\mu_c$. 

However, when a qudit undergoes a $Z$ error, then, depending on the vertex type, different syndromes are induced on the faces.
Suppose that the vertex on which operator $Z$ acts belongs to $\mathsf{V}_o$.
The stabilizers of the faces have operator $X$ on vertices belonging to $\mathsf{V}_e$.
From Eq.~\eqref{eq:qudit-XZ-commutation}, we see that the syndrome is $\omega^{1}$, equivalently $1$
on the three faces incident on the vertex. 
The excitations produced by $Z$ type errors are electric charges.
The charge induced on a face of color $c$ by a $Z$ error is denoted $\epsilon_c$, see Fig.~\ref{fig:Qd-cc-Z-error-syndrome-minus}.
Similarly, if the error operator $Z$ acts on vertex belonging to $\mathsf{V}_e$, then charge $\epsilon_c^{-1}$ is induced on a face of color $c$, see Fig.~\ref{fig:Qd-cc-Z-error-syndrome-plus}.

\begin{figure}[htb]
    \centering
    \begin{subfigure}{.225\textwidth}
        \centering
        \includegraphics[scale = 1.1]{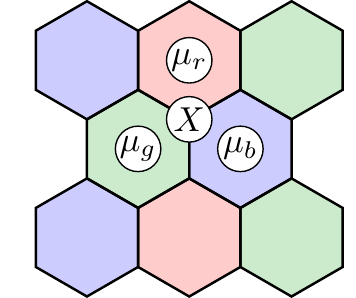}
        \subcaption{}
    \label{fig:Qd-cc-X-error-syndrome}
    \end{subfigure}
    ~
     \begin{subfigure}{.225\textwidth}
        \centering
        \includegraphics[scale = 1.1]{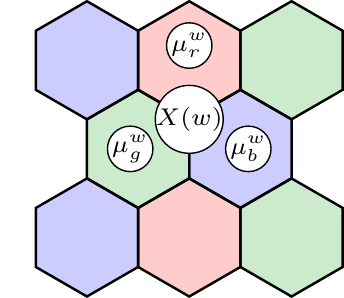}
        \subcaption{}
    \label{fig:Qd-cc-weighted-pauli}
    \end{subfigure}
     \begin{subfigure}{.225\textwidth}
        \centering
        \includegraphics[scale = 1.1]{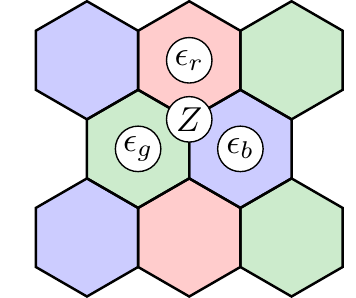}
        \subcaption{}
    \label{fig:Qd-cc-Z-error-syndrome-minus}
    \end{subfigure}
    ~
  \begin{subfigure}{.225\textwidth}
        \centering
        \includegraphics[scale = 1.1]{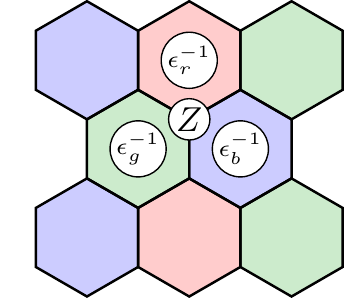}
        \subcaption{}
    \label{fig:Qd-cc-Z-error-syndrome-plus}
    \end{subfigure}
    \caption{ Syndromes created on faces as a result of $X$ and $Z$ error. (a) An $X$ error on a qudit violates $Z$ stabilizers on three faces. Hence, three syndromes are created. (b) The syndromes induced by operator $X(w)$. (c) A $Z$ error acting on a vertex $u \in \mathsf{V}_e$. Note that the syndromes created on the faces incident on the vertex are $\epsilon_c^{-1}$. (d) When a $Z$ error acts on a vertex $u \in \mathsf{V}_o$, the syndromes $\epsilon_c$ are created on faces incident on the vertex.
    }
    \label{fig:Qd-cc-pauli}
\end{figure}

In case of qudits, we also need to represent powers of operators $Z$ and $X$.
This is done by indicating the power along with the operator, see Fig.~\ref{fig:Qd-cc-weighted-pauli}.
If the error operator $X(w)$ acts on a vertex, it induces syndromes $\mu_c^w$ on faces of color $c$.
When the operator $Z(w)$ acts on a vertex, it induces either $\epsilon_c^{-w}$ or $\epsilon_c^{w}$ depending on whether the vertex belongs to $\mathsf{V}_e$ or $\mathsf{V}_o$ respectively.

Excitations in color code are characterized by the color of face on which the excitation lives on and the error that caused it.
The set of excitations in the qudit color code is given below:
\begin{equation}
\mathcal{B} = \{ \mu_c^m \epsilon_c^n, c \in \{r,g,b \}, m,n \in \{0,1,\dots,d-1 \} \}.
\label{eqn:syndromes}
\end{equation}
Totally, there are $3d^2 - 1$ possible nontrivial charges.
However, not all charges are independent; only $2d^2$ of them are independent~\cite{AloshiousSarvepalli2019}.

Suppose we have a multi-qudit error on a face. 
Then each of these errors produces a charge on that face. 
We can only observe the effective charge on that face obtained by combining all the individual charges. 
This brings forth the question of how to combine these charges. 
The symbol $\times$ is used to indicate fusion of charges.
The fusion rules for syndromes in qudit color codes can be obtained by considering the commutation of errors 
causing the charges with the respective stabilizer generators. 
These are given in Table.~\ref{tab:fusion-rules}.
They were also stated in Ref.~\cite{AloshiousSarvepalli2019}.

\begin{table}[htb]
    \centering
    \begin{tabular}{lll}
         \hline
         \hline
         $\epsilon_c^m \times \epsilon_c^n$ & $=$ & $\epsilon_c^{m + n}$,\\
         $\mu_c^m \times \mu_c^n$ & $=$ & $\mu_c^{m + n}$,\\
         $\mu_c^m \times \epsilon_c^n$ & $=$ & $\mu_c^{m}\epsilon_c^n$,\\
         $\mu_c^m \times \mu_{c^\prime}^m$ & $=$ & $\mu_{c^{\prime\prime}}^{-m}$,\\
         $\epsilon_c^m \times \epsilon_{c^\prime}^m$ & $=$ & $\epsilon_{c^{\prime\prime}}^{-m}$.\\
         \hline
         \hline
    \end{tabular}
    \caption{Fusion rules for excitations in qudit color code.}
    \label{tab:fusion-rules}
\end{table}
Apart from the excitations listed in Equation~\eqref{eqn:syndromes}, we may also encounter excitations of the form $\mu_c^m \times \epsilon_{c^\prime}^n$ where $c \ne c^\prime$.

\noindent \emph{Syndrome movement.} 
Syndrome on a face can be moved to another by application of suitable operators.
The operator $Z_u Z_v$ as shown in Fig.~\ref{fig:Qd-cc-string-pauli-1} creates the syndromes shown in Fig.~\ref{fig:Qd-cc-string-pauli-2}.
Note that the syndromes on blue and green faces vanish leaving only syndromes on red faces.
Suppose that a red face hosts a syndrome $\epsilon_r$ as shown in Fig.~\ref{fig:Qd-cc-syndrome-movement-x}.
Applying operator $Z_u Z_v$ moves the syndrome on the left red face to the right red face.
Similarly, $Z$ syndromes can be moved around by applying operator $X_u^\dagger X_v$, see Fig.~\ref{fig:Qd-cc-string-pauli-5}, Fig.~\ref{fig:Qd-cc-string-pauli-6} and Fig.~\ref{fig:Qd-cc-syndrome-movement-z}.
Such operators that move a syndrome from one face to another are called \textit{hopping operators}~\cite{bombin2012universal,BhagojiSarvepalli2015,AloshiousSarvepalli2019}.
Note that the hopping operators do not commute only with the stabilizers of faces between which excitations are moved and commute with the rest of stabilizers.

\begin{figure}[htb]
    \centering
    \begin{subfigure}{.225\textwidth}
        \centering
        \includegraphics[scale = .85]{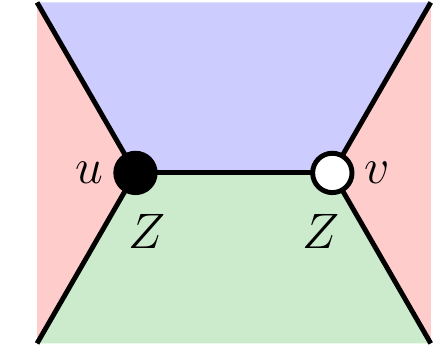}
        \subcaption{}
        \label{fig:Qd-cc-string-pauli-1}
    \end{subfigure}
    ~
    \begin{subfigure}{.225\textwidth}
        \centering
        \includegraphics[scale = .85]{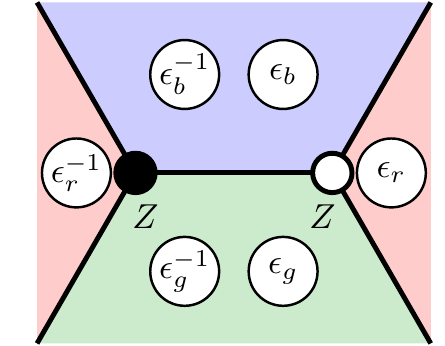}
        \subcaption{}
        \label{fig:Qd-cc-string-pauli-2}
    \end{subfigure}
    ~
    \begin{subfigure}{.225\textwidth}
        \centering
        \includegraphics[scale = .85]{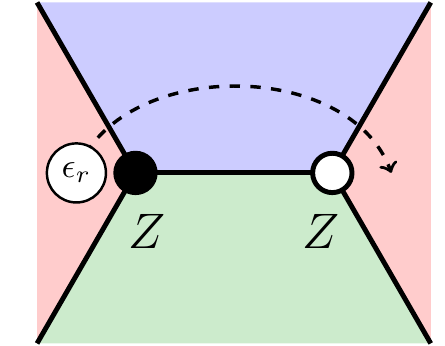}
        \subcaption{}
        \label{fig:Qd-cc-syndrome-movement-x}
    \end{subfigure}
    ~
    \begin{subfigure}{.225\textwidth}
        \centering
        \includegraphics[scale = .85]{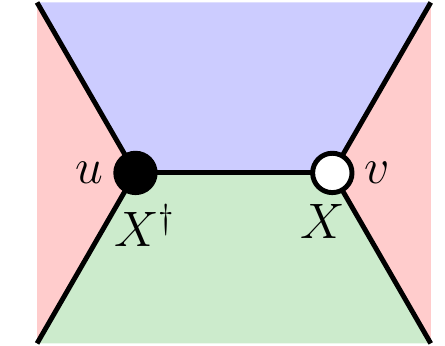}
        \subcaption{}
        \label{fig:Qd-cc-string-pauli-5}
    \end{subfigure}
    ~
    \begin{subfigure}{.225\textwidth}
        \centering
        \includegraphics[scale = .85]{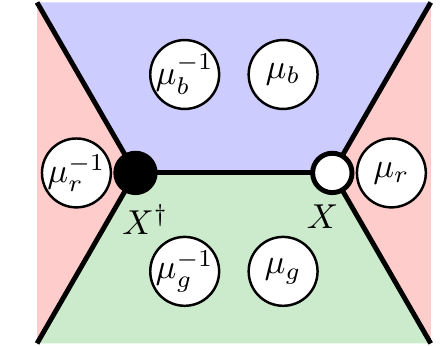}
        \subcaption{}
        \label{fig:Qd-cc-string-pauli-6}
    \end{subfigure}
     ~
    \begin{subfigure}{.225\textwidth}
        \centering
        \includegraphics[scale = .85]{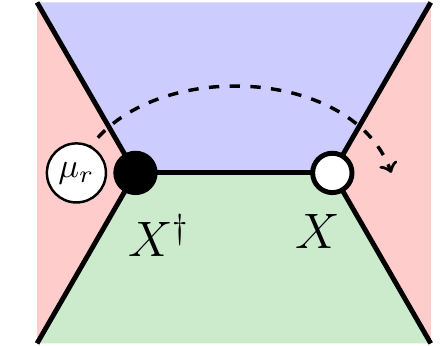}
        \subcaption{}
        \label{fig:Qd-cc-syndrome-movement-z}
    \end{subfigure}
    \caption{Moving the excitations, equivalently, syndromes. (a) Suppose that we have $Z$ operator on two vertices as shown. One of the vertices is even and the other is odd. (b) Each $Z$ operator violates $X$ type stabilizer on the faces incident on the vertex. The syndromes on the faces are as shown. The $Z$ operators shown in (a) can move a syndrome from one face to another as shown in (c) . (d) $X$ operator on two two neighboring vertices. (e) Each $X$ operator introduces syndromes on three faces as shown. (f) Such two qudit operators can move syndromes from one face to another.}
    \label{fig:syndrome_movement}
\end{figure}

\section{charge-and-color-permuting twists}
\label{sec:charge-color}
In this section, we present twists that permute both charge and color of syndromes.
We give a systematic construction of these twists.
There are two parts to this construction: lattice modification and stabilizer assignment.
The procedure given here works for arbitrary $2$-colexes.
We restrict to hexagonal lattice in this paper, but these ideas can be extended to other lattices.

A twist is a face in the lattice that permutes a label of an excitation when that excitation  is moved around it.
A charge-and-color-permuting twist permutes both charge and color label of the excitation.
This exchange is denoted as $\longleftrightarrow$. 
So $e \longleftrightarrow f$ indicates that $e$ and $f$ are exchanged. 
A formal definition is as follows:
\begin{definition}[Charge-and-color-permuting twist]
A charge-and-color-permuting twist of color $c$ has the following action on the excitations when they are moved around it:
\begin{subequations}
\begin{eqnarray*}
\mu_{c}^m &\longleftrightarrow& \epsilon_{c}^m, \\
\mu_{c^\prime}^m  &\longleftrightarrow& \epsilon_{c^{\prime \prime}}^m , \\
\mu_{c^{\prime \prime}}^m  &\longleftrightarrow& \epsilon_{c^\prime}^m,
\end{eqnarray*}
\end{subequations}
where $c$, $c^{\prime} $ and $c^{\prime \prime}$ are all distinct colors, $m = 1,\dots,d-1$
\end{definition}

Note that a charge-and-color-permuting twist of color $c$ permutes only the charge label of an excitation with color label $c$ and permutes both color and charge label for other excitations.
A charge-and-color-permuting twist of color $c$ leaves the following excitations unchanged: $\mu_{c}^\alpha \epsilon_{c}^\alpha$, $\mu_{c^{\prime}}^\alpha \times \epsilon_{c^{\prime \prime}}^\alpha$ and $\mu_{c^{\prime \prime}}^\alpha \times \epsilon_{c^{\prime}}^\alpha$ where $\alpha = 1,2,\dots, d-1$ and the colors $c$, $c^{\prime}$ and $c^{\prime \prime}$ are all distinct.

\subsection{Construction of color codes with charge-and-color-permuting twists}
The introduction of charge-and-color-permuting twists into a color code lattice has 
two parts: lattice modification and stabilizer assignment.
First observe that the charges are moved along the edges from one face to another face. 
In a 2-colex, all the outgoing edges of a face of color $c$ connect it to faces of same color. 
So if we want to introduce color permutation of the charges, we need edges between faces of different color. 
Naturally, this implies a lattice modification. 
We take the following approach. This is similar to the qubit color permuting twists~\cite{GowdaSarvepalli2021}.
The main difference is in the stabilizer assignment step. 
We summarize the procedure below:
\begin{compactenum}[i)]
\item Choose an edge $e = (p,q)$ of color $c \in \{r, g, b \}$. The edge $e$ is common to faces of colors $c^\prime \neq c$ and $c^{\prime \prime} \neq c$. 

\item Connect the neighbors of vertex $p$ with an edge and remove vertex $p$. 
Also, repeat this procedure for vertex $q$.

\item This results in two faces (which were connected by the edge $e = (p,q)$) having an odd number of edges and a face with an even number of edges.
The faces with an odd number of edges are colored $c$ and the face with an even number of edges is colored either $c^\prime$ or $c^{\prime \prime}$.
\end{compactenum}
\medskip

This procedure destroys the local three colorability of the lattice.
As a result, there exists adjacent faces of the same color.
The common edge to such faces form a path in the shrunk lattice of appropriate color.
For instance, in Fig.~\ref{fig:Qd-cc-charge-color-int-4}, this path is formed in the red shrunk lattice.
The sequence of such edges is called $T$-line~\cite{Bombin2011, GowdaSarvepalli2020, GowdaSarvepalli2021}.
The proof for the existence of $T$-lines is similar to that of color codes with color permuting twists~\cite{GowdaSarvepalli2021}.
{We also assume that the twists are sufficiently far apart so that no two $T$-lines cross each other.}
The virtual path in the lattice that marks the point where the syndromes are permuted is called the domain wall.
Note that an alternative definition for twists can be given with reference to the domain wall:
twists are the faces in which the domain wall terminates.
In Fig.~\ref{fig:Qd-cc-charge-color-int-4}, the domain wall is indicated as a dashed line.
The number of charge-and-color-permuting twists is always an even number. Reasoning is similar to the parity of color permuting twists in qubit color codes~\cite{GowdaSarvepalli2021}.

\begin{figure}[htb]
    \centering
        \includegraphics[scale = 1]{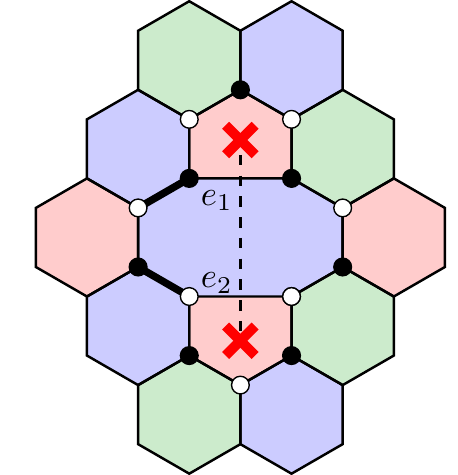}
     \caption{The lattice after modification is no longer bipartite. The twist faces are indicated by a cross and the domain wall is indicated as dashed line terminating in twist faces. Also note that there are now adjacent faces of the same color and the common edges to them form a path in the red shrunk lattice.}
    \label{fig:Qd-cc-charge-color-int-4}
\end{figure}

After the introduction of twists,  the modified lattice is no longer bipartite. 
However, we retain the partitioning of vertices as in  the parent lattice i.e. the lattice before introducing twists.
Now  there are neighboring vertices belonging to the same partition, see Fig.~\ref{fig:Qd-cc-charge-color-int-4}.
Specifically, the common vertices to a twist and a modified face adjacent to it are of the same type, see Fig.~\ref{fig:Qd-cc-charge-color-int-4}.

We can also identify two types of edges in the lattice namely, the edges that connect vertices from different bipartition ($\mathsf{E}_d$) and edges like $e_1$ and $e_2$ in Fig.~\ref{fig:Qd-cc-charge-color-int-4} that connect the vertices from the same bipartition ($\mathsf{E}_s$).
It can be seen that edges in $\mathsf{E}_d$ are incident on faces of the same color whereas edges from $\mathsf{E}_s$ are incident on faces of different color.
Also, the edge in $\mathsf{E}_s$ form the common edges to a twist and a modified face and two adjacent modified faces.
The domain wall cuts across all these edges in the lattice.

Based on the type of edges, we can identify three types of faces in the lattice:
\begin{compactenum}[i)]
\item Normal Faces ($\mathcal{D}_0$): 
These faces do not contain edges from $\mathsf{E}_s$ and have an even number of edges.
Domain wall does not pass through these faces.

\item Modified Faces ($\mathcal{D}_2$): 
These faces contain exactly two edges from $\mathsf{E}_s$ and also have an even number of edges.
Domain wall passes through these faces.
An example of such face is the blue octagon face in Fig.~\ref{fig:Qd-cc-charge-color-int-4}.

\item Twists ($\mathcal{T}$): 
These faces contain exactly an edge from $\mathsf{E}_s$ and have an odd number of edges, see the red pentagon faces in Fig.~\ref{fig:Qd-cc-charge-color-int-4}.
Domain wall terminates in these faces.
\end{compactenum}

\medskip
We next move on to assign stabilizer generators to the three types of faces described above.

\subsubsection{Stabilizer assignment}
We now present stabilizer assignment to each of the aforesaid class of faces.
On normal faces ($\mathcal{D}_0$), the stabilizers are defined as given in Equations~\eqref{eqn:qudit-stab}.
To define stabilizers on a modified face $m$, we first partition the vertices of the face $m$ into two sets: the vertices shared with a $T$-line, $Q_1 = V(m) \cap V_T$ where $V_T$ is the set of all vertices in the support of $T$-lines and the vertices not shared with a $T$-line, $Q_2 = V(m) \setminus Q_1$. 
The stabilizers defined on a modified face $m$ are as follows:
\begin{subequations}
\begin{eqnarray}
B_{m, 1} &=& \prod_{v \in Q_1} Z_v  \prod_{v \in Q_2 \cap \mathsf{V}_e} X_v  \prod_{v \in Q_2 \cap \mathsf{V}_o} X_v^\dagger,\\
B_{m, 2} &=&   \prod_{v \in Q_1 \cap \mathsf{V}_e} X_v  \prod_{v \in Q_1 \cap \mathsf{V}_o} X_v^\dagger \prod_{v \in Q_2} Z_v.
\end{eqnarray}
\label{eqn:qudit-cc-modified-stab}
\end{subequations}

On a twist face $\tau$, only one stabilizer is defined  as given below:
\begin{equation}
    B_\tau = \prod_{v \in V(\tau) \cap \mathsf{V}_e} Z_v X_v \prod_{v \in V(\tau) \cap \mathsf{V}_o} Z_v X_v^\dagger.
    \label{eqn:qudit-cc-twist-stab}
\end{equation}

\begin{figure}[htb]
    \centering
    \begin{subfigure}{.225\textwidth}
        \centering
        \includegraphics[scale = .75]{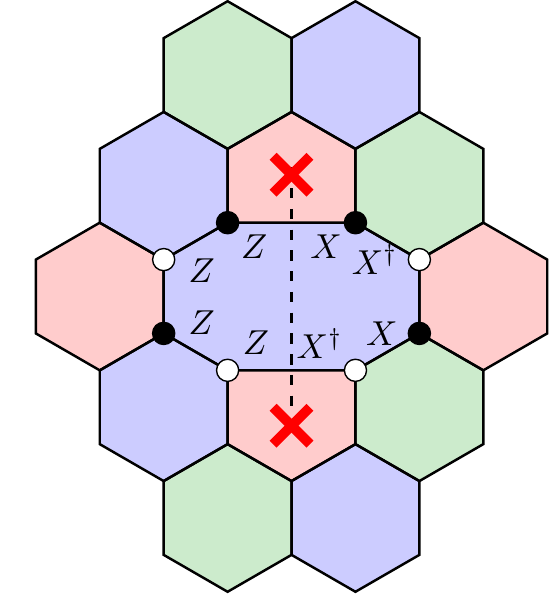}
        \subcaption{}
        \label{fig:Qd-cc-charge-color-domain-wall-stab-1}
    \end{subfigure}
    ~
    \begin{subfigure}{.225\textwidth}
        \centering
        \includegraphics[scale = .75]{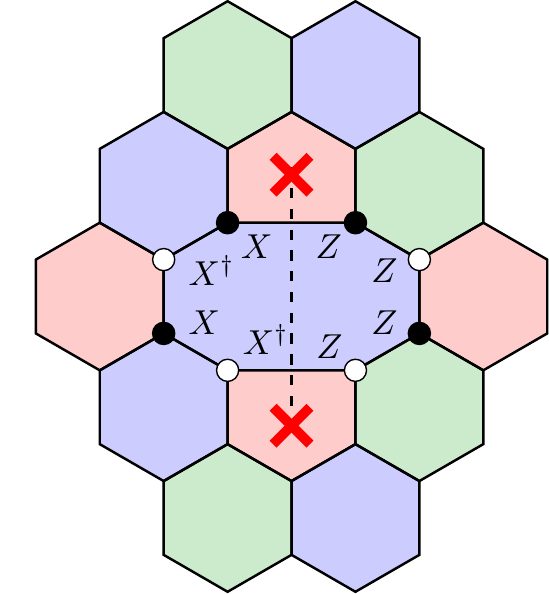}
        \subcaption{}
        \label{fig:Qd-cc-charge-color-domain-wall-stab-2}
    \end{subfigure}
    ~
    \begin{subfigure}{.225\textwidth}
        \centering
        \includegraphics[scale = .75]{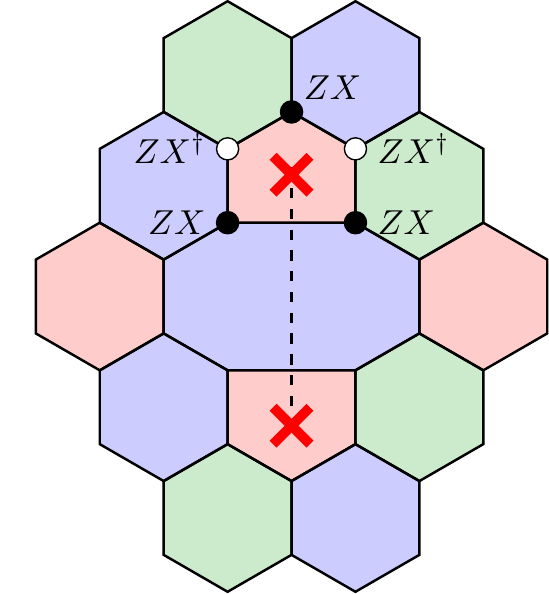}
        \subcaption{}
        \label{fig:Qd-cc-charge-color-permuting-twist-stab}
    \end{subfigure}
    \caption{Stabilizer generators on modified faces and twist. (a) A stabilizer generator with $Z$ on vertices to the left of the domain wall and $X$ to the right. (b) Another stabilizer generator with $X$ on left of the domain wall and $Z$ to the right. (c) Stabilizer on the twist face.}
    \label{fig:Qd-cc-charge-color-domain-wall-stab}
\end{figure}

The stabilizers defined on modified faces are shown in Fig.~\ref{fig:Qd-cc-charge-color-domain-wall-stab-1} and Fig.~\ref{fig:Qd-cc-charge-color-domain-wall-stab-2} and the stabilizer defined on twist is shown in Fig.~\ref{fig:Qd-cc-charge-color-permuting-twist-stab}. 
We next show that the stabilizers defined above commute.

\begin{lemma}[Stabilizer commutation.]
The stabilizer generators defined in Equations~\eqref{eqn:qudit-stab}, ~\eqref{eqn:qudit-cc-modified-stab},~\eqref{eqn:qudit-cc-twist-stab} commute.
\label{lm:stabilizer-commutation}
\end{lemma}

Note that any two adjacent faces share exactly an edge.
If the generalized Pauli operators corresponding to stabilizer generators are same on the common vertices, then clearly they commute. 
If the generalized Pauli operators are different, then it can be verified that on the common vertices, the operators are such that the stabilizer generators commute, see Fig.~\ref{fig:Qd-cc-stab-commutation}.
Stabilizers defined on the same face commute as the phase resulting from the exchange of operators corresponding to the two stabilizer generators on any two adjacent vertices is zero.
Since all faces except twists are even cycles, the phase resulting from exchange of operators is zero.
A detailed proof is given in Appendix~\ref{sec:stab-commutation}.

\begin{figure*}
    \centering
    \begin{subfigure}{.225\textwidth}
        \centering
        \includegraphics[scale = .75]{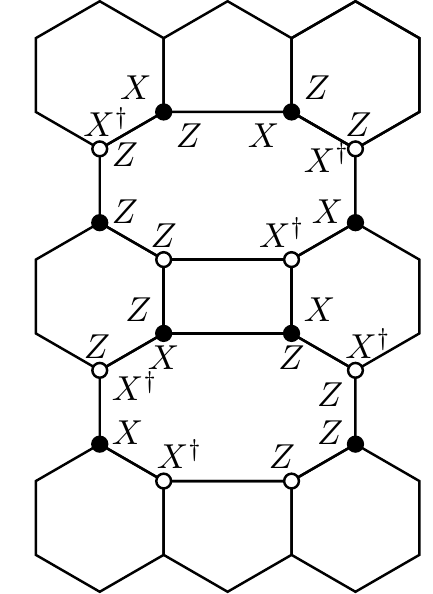}
        \subcaption{}
        \label{fig:Qd-cc-stab-commutation-1}
    \end{subfigure}
    ~
    \begin{subfigure}{.225\textwidth}
        \centering
        \includegraphics[scale = .65]{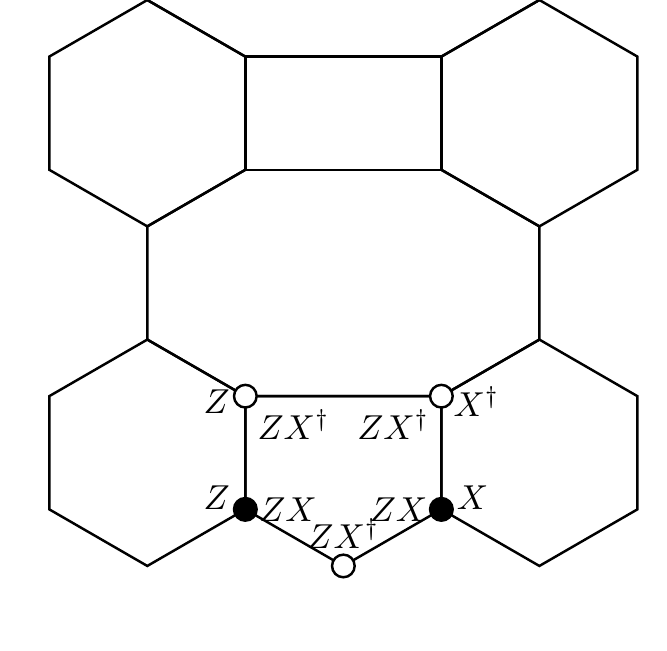}
        \subcaption{}
        \label{fig:Qd-cc-stab-commutation-2}
    \end{subfigure}
    ~
    \begin{subfigure}{.225\textwidth}
        \centering
        \includegraphics[scale = .75]{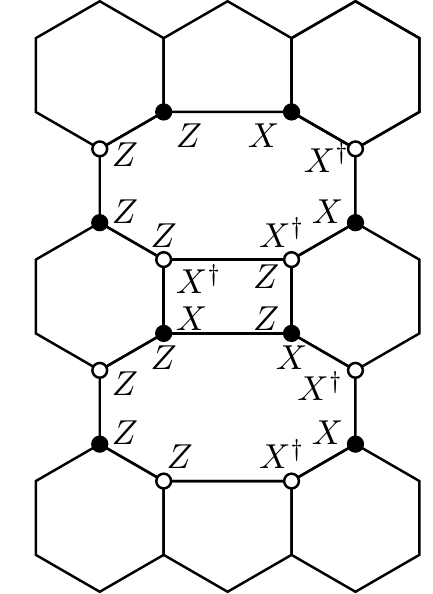}
        \subcaption{}
        \label{fig:Qd-cc-stab-commutation-3}
    \end{subfigure}
     ~
    \begin{subfigure}{.225\textwidth}
        \centering
        \includegraphics[scale = .75]{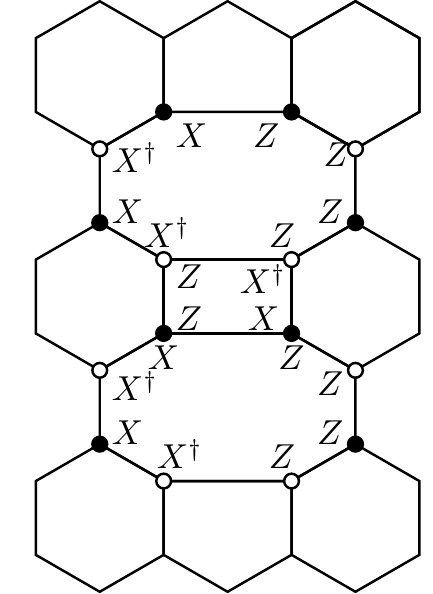}
        \subcaption{}
        \label{fig:Qd-cc-stab-commutation-5}
    \end{subfigure}
    ~
    \begin{subfigure}{.3\textwidth}
        \centering
        \includegraphics[scale = .65]{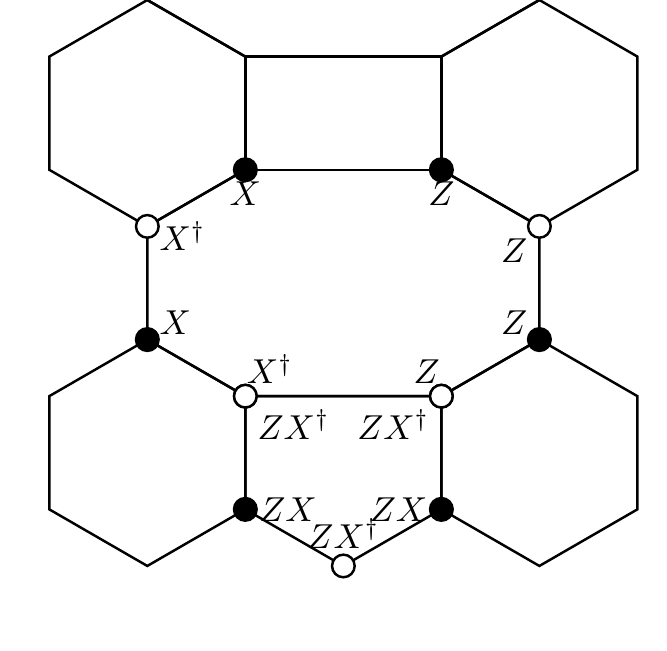}
        \subcaption{}
        \label{fig:Qd-cc-stab-commutation-4}
    \end{subfigure}
     ~
    \begin{subfigure}{.3\textwidth}
        \centering
        \includegraphics[scale = .65]{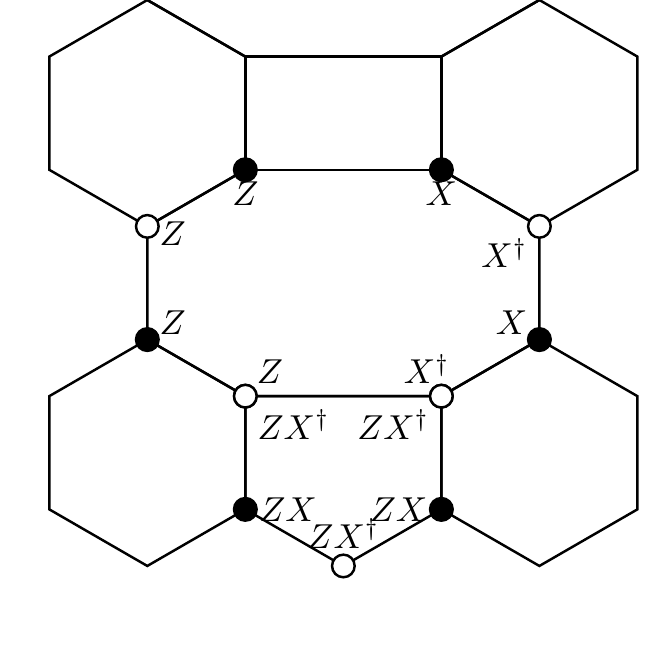}
        \subcaption{}
        \label{fig:Qd-cc-stab-commutation-6}
    \end{subfigure}
    ~
    \begin{subfigure}{.3\textwidth}
        \centering
        \includegraphics[scale = .75]{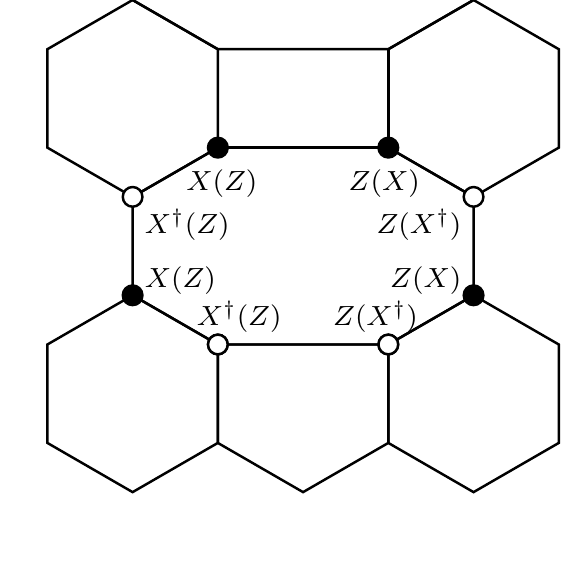}
        \subcaption{}
        \label{fig:Qd-cc-stab-commutation-7}
    \end{subfigure}
    \caption{Commutation of stabilizers defined on a) modified and unmodified faces  b) twist and unmodified face c) - d) two distinct adjacent modified faces and e) - f) modified face and twist. We have only considered cases where the generalized Pauli operators of the stabilizer generators on the common vertices differ. (g) Stabilizer generators defined on the same face.}
    \label{fig:Qd-cc-stab-commutation}
\end{figure*}

With lattice modification and stabilizer assignment in place, we show how excitations are permuted as they are moved across the domain wall.
The sublattice before lattice modification is shown in Fig.~\ref{fig:Qd-cc-charge-color-int-0}.
Twists are introduced by breaking the local three colorability of the sublattice as shown in Fig.~\ref{fig:Qd-cc-charge-color-int-1}.
Suppose that a blue face carries an excitation $\mu_b$ shown in Fig.\ref{fig:Qd-cc-charge-color-int-1}.
We next the apply the operator $X^\dagger_u Z^\dagger_v$, see Fig.~\ref{fig:Qd-cc-charge-color-int-2}.
This operator commutes with the stabilizer generators defined on twist and modified face.
Hence, no excitations are induced on twist and the blue modified face.
Note that the operator $X^\dagger_u Z^\dagger_v$ annihilates the excitation the blue face and creates an excitation on the green face to the right of the domain wall, see Fig.~\ref{fig:Qd-cc-charge-color-int-3}.
Effectively, this operation has moved the excitation $\mu_b$ across the domain wall and permuted it to $\epsilon_g$.
Thus, the required permutation of excitations is achieved with the lattice modification and the stabilizer assignment. 

\begin{figure}[htb]
    \centering
    \begin{subfigure}{.225\textwidth}
        \centering
        \includegraphics[scale = .85]{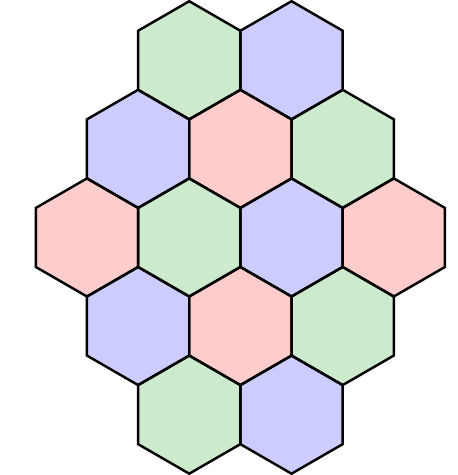}
        \subcaption{}
        \label{fig:Qd-cc-charge-color-int-0}
    \end{subfigure}
    ~
    \begin{subfigure}{.225\textwidth}
        \centering
        \includegraphics[scale = .85]{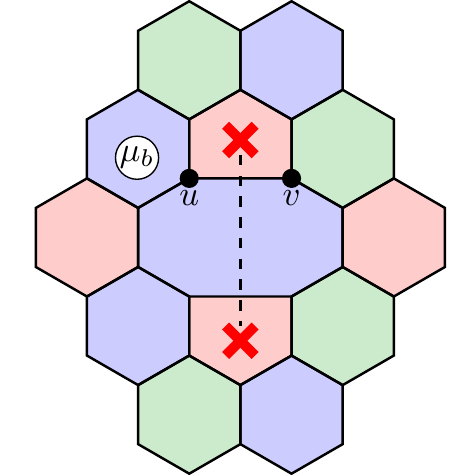}
        \subcaption{}
        \label{fig:Qd-cc-charge-color-int-1}
    \end{subfigure}
    ~
    \begin{subfigure}{.225\textwidth}
        \centering
        \includegraphics[scale = .95]{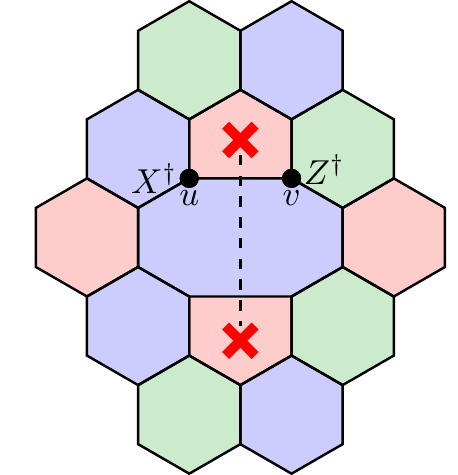}
        \subcaption{}
        \label{fig:Qd-cc-charge-color-int-2}
    \end{subfigure}
     ~
    \begin{subfigure}{.225\textwidth}
        \centering
        \includegraphics[scale = .95]{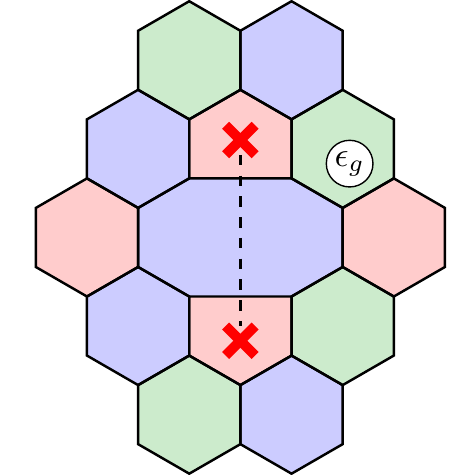}
        \subcaption{ }
        \label{fig:Qd-cc-charge-color-int-3}
    \end{subfigure}
    \caption{Illustration of charge and color permutation induced by the twist. (a) A sublattice of a $2$-colex. (b) Twists are created by altering the sublattice as shown. We would like to move the charge $\mu_b$  on the blue face across the domain wall. (c) The operator $X^\dagger_u Z^\dagger_v$ commutes with the stabilizers on twist and modified face. (d) The operator $X^\dagger_u Z^\dagger_v$ annihilates the charge on the blue face and creates the charge $\epsilon_g$ on the green face.}
    \label{fig:Qd-cc-charge-color-demo}
\end{figure}

The stabilizer generators described are not all independent.
The dependency among the stabilizer generators is explored next.

\subsubsection{Stabilizer constraint}
In the most general case, a lattice can have twists of different colors.
Giving stabilizer constraints for such case where all twists are not of the same color is difficult.
So we restrict to the case where all twists in the lattice are of red color and modified faces are colored blue or red. 
(Blue modified faces are created during twist introduction, while the red modified faces are obtained during twist movement wherein the old twist becomes the new modified face).
The stabilizers defined in Equations~\eqref{eqn:qudit-stab}, ~\eqref{eqn:qudit-cc-modified-stab},~\eqref{eqn:qudit-cc-twist-stab} satisfy the following constraint: 
\begin{equation}
    \prod_{m \in \mathsf{F}_b \cap \mathcal{D}_2} B_{m,1} B_{m,2} \prod_{f \in \mathsf{F}_b \setminus \mathcal{D}_2} B_{f}^X B_{f}^Z  \prod_{f \in \mathsf{F}_g} {B_{f}^X}^\dagger {B_{f}^Z}^\dagger =  W^2,
    \label{eqn:stab-constr-1}
\end{equation}
where $\mathcal{D}_2$ is the set of all modified faces and $W$ is the operator as defined below:
\begin{equation*}
    W =  \prod_{v \in V_t \cap \mathsf{V}_e} Z_v X_v \prod_{v \in V_W \cap \mathsf{V}_o} Z_v {X_v^\dagger}
\end{equation*}
where $V_W$ is the set of vertices in the support of all the $T$-lines.
The constraint in Equation~\eqref{eqn:stab-constr-1} tells that the product of all stabilizer generators defined on blue faces (modified and unmodified) and the conjugate of the stabilizer generators defined on the green faces is an operator with support on all the $T$-lines.
Note that  only the vertices of the red faces, either a modified face or a twist, not shared with a green face are along the $T$-line.
Therefore, $W$ can be expressed as product of red and green face stabilizers as below:

\begin{eqnarray*}
W &=& \prod_{\tau \in \mathcal{T}} B_{\tau} \prod_{m \in \mathsf{F}_r \cap \mathcal{D}_2} B_{m,1} B_{m,2} \prod_{f \in \mathsf{F}_r \setminus \{ \mathcal{D}_2 \cup \mathcal{T} \}} B_{f}^X B_{f}^Z\\
& & \times \prod_{f \in \mathsf{F}_g} \!\!\!{B_{f}^X}^\dagger {B_{f}^Z}^\dagger
\end{eqnarray*}
where $\mathcal{T}$ is the set of all twist faces.
Combining the above equation and Equation~\eqref{eqn:stab-constr-1}, we get,
\begin{small}
\begin{equation}
\begin{split}
    & \prod_{m \in \mathsf{F}_b \cap \mathcal{D}_2} B_{m,1} B_{m,2} \prod_{f \in \mathsf{F}_b \setminus \mathcal{D}_2} B_{f}^X B_{f}^Z  \prod_{f \in \mathsf{F}_g} {B_{f}^X}^\dagger {B_{f}^Z}^\dagger = \\   
    & \left( \prod_{\tau \in \mathcal{T}} B_{\tau} \prod_{m \in \mathsf{F}_r \cap \mathcal{D}_2} \!\!\! B_{m,1} B_{m,2}  \!\!\!\prod_{f \in \mathsf{F}_r \setminus \{ \mathcal{D}_2 \cup \mathcal{T} \}} B_{f}^X B_{f}^Z \prod_{f \in \mathsf{F}_g} {B_{f}^X}^\dagger {B_{f}^Z}^\dagger \right)^2.
\end{split}
    \label{eqn:stabilizer-constraint}
\end{equation}
\end{small}

The above equation indicates the presence of a dependent stabilizer.
We take the dependent stabilizer to be one of the stabilizer generators on the unbounded blue face.
This makes the other stabilizer on the unbounded face independent.
The independent stabilizer is nonlocal and has to be measured during error correction which is undesirable.
Note that there is a stabilizer generator of the form 
\begin{equation}
    B_{f_e} = \prod_{v \in V(f_e) \cap \mathsf{V}_e} Z_v X_v \prod_{v \in V(f_e) \cap \mathsf{V}_o} Z_v X_v^\dagger.
    \label{eqn:external-face-stabilizer}
\end{equation}
on the blue unbounded face.
This stabilizer generator is the product of $Z$ and $X$ type stabilizer generators.
We choose this to be the independent stabilizer.
Note that the stabilizer defined in Equation~\eqref{eqn:external-face-stabilizer} satisfies the constraint in Equation~\eqref{eqn:stabilizer-constraint}.
Also, note that the stabilizer generator $B_{f_e}$ can be expressed as a combination of other face stabilizers and is therefore dependent. 
With the complete set of stabilizer generators and the constraints they satisfy, we now proceed to derive the number of encoded qudits in lattices with $t$ charge-and-color-permuting twists.

\begin{theorem}[Encoded qudits]
A qudit color code lattice with $t$ charge-and-color-permuting twists encodes $t - 2$ logical qudits.
\end{theorem}

\begin{proof}
All vertices in the lattice are trivalent and hence we get $2e = 3v$ where $e$ and $v$ are the number of edges and vertices in the lattice respectively.
Using this in the Euler formula for a two-dimensional plane, $v + f - e = 2$ where $f$ is the number of faces in the lattice, we get, $2f = v + 4$.
The number of stabilizers is $2f - t - 1$ since we define only one stabilizer on twists and external unbounded blue face.
However, the stabilizer on the unbounded blue face is dependent, see Equation~\eqref{eqn:stabilizer-constraint}.
Hence the number of independent stabilizers is $s = 2f - t -2 = v - (t - 2)$.
Hence, the number of logical qudits is $t - 2 = 2\left(\frac{t}{2} - 1 \right)$.
\end{proof}

\begin{remark}
The number of encoded qudits here is twice that of qudit surface codes with twists~\cite{GowdaSarvepalli2020}.
\end{remark}

\section{Mapping Generalized Pauli operators to Strings}
\label{sec:pauli-string-mapping}
In this section, we present a mapping between generalized Pauli operators and strings in the presence of charge-and-color-permuting twists.
We represent the stabilizer generators and logical operators in the string notation that we develop and give the canonical form of logical operators which will be used while implementing encoded gates with twists.
It is helpful to hide the lattice information and represent the operators as strings.
We need to represent $X$ and $Z$ operators for which we use different types of strings.
A dashed string of any color represents the $X$ operator, see Fig.~\ref{fig:Qd-cc-X-error-string} and a solid string of any color represents the $Z$ operator, see Fig.~\ref{fig:Qd-cc-Z-error-string}.
The string corresponding to an $X$ error terminates on three faces. 
The end points of the string correspond to nonzero syndromes (nontrivial excitations).

In qudit codes, we have to represent powers of generalized Pauli operators.
We use vertex weights to indicate powers of Pauli operator.
The operators $X(w)$ and $Z(w)$ are represented as dashed and solid strings with vertex weight $w$, see Fig.~\ref{fig:Qd-cc-X-error-weighted-string} and Fig.~\ref{fig:Qd-cc-Z-error-weighted-string}.
The excitations at the end points of string corresponding to $X(w)$  are obtained by fusing the excitations 
for the strings $X$, $w$ times. 
So for instance, in this case the end points carry the excitations $\mu_c^{w}$.
We do not explicitly represent these excitations in the string representation. 
They can be found from the  weight of the vertex on which the error occurs. 
Using these strings as the building blocks, we present the string representation for multi-qudit operators.
In this paper, we are primarily interested in operators of the form shown in Fig.~\ref{fig:syndrome_movement}.
Such operators are the building blocks of stabilizer generators and logical operators.

\begin{figure}[htb]
    \centering
    \begin{subfigure}{.225\textwidth}
        \centering
        \includegraphics[scale = 1]{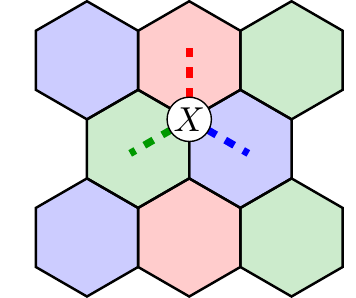}
        \subcaption{}
    \label{fig:Qd-cc-X-error-string}
    \end{subfigure}
    ~
    \begin{subfigure}{.225\textwidth}
        \centering
        \includegraphics[scale = 1]{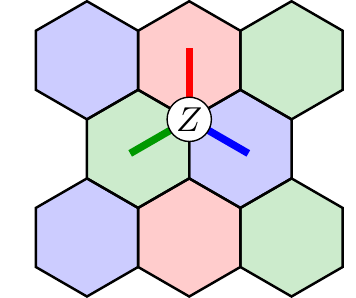}
        \subcaption{}
    \label{fig:Qd-cc-Z-error-string}
    \end{subfigure}
    ~
    \begin{subfigure}{.225\textwidth}
        \centering
        \includegraphics[scale = 1]{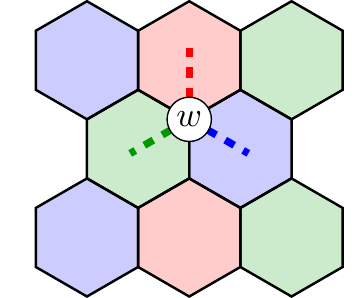}
        \subcaption{}
    \label{fig:Qd-cc-X-error-weighted-string}
    \end{subfigure}
    ~
    \begin{subfigure}{.225\textwidth}
        \centering
        \includegraphics[scale = 1]{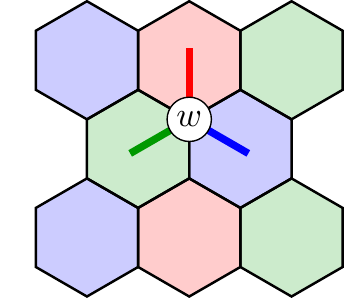}
        \subcaption{}
    \label{fig:Qd-cc-Z-error-weighted-string}
    \end{subfigure}
    \caption{Mapping generalized Pauli operators to strings  in a $2$-colex. (a) We use dashed strings to represent an $X$ error. The strings are open ended in the faces as they have a syndrome. (b) A $Z$ error is represented by a solid string irrespective of whether it acts on vertices in $\mathsf{V}_e$ or $\mathsf{V}_o$. (c)-(d) String representation of $X(w)$ and $Z(w)$ operators respectively.
    }
    \label{fig:Qd-cc-error-strings}
\end{figure}

\subsection{Generalized Pauli operators as strings in a 2-colex}
In general when we combine the elementary strings for $X$ and $Z$ errors on single qudits, we can obtain the string representation for an arbitrary error. 
A subset of the errors and their representations are more useful and adequate for our purposes. 
We are primarily interested in the case where the end of the strings of two distinct errors can be merged.
This is the case when the end points of two strings carry excitations which fuse to the vacuum. 
Equivalently, the syndromes produced on these errors add up to zero.
It is possible to consider the case when the excitations do not fuse to the vacuum by allowing for weighted edges.

We begin by considering the operator $Z_u Z_v$ shown in Fig.~\ref{fig:Qd-cc-string-pauli-1} where $u \in \mathsf{V}_e$ and $v \in \mathsf{V}_o$. 
Using the string representation of the $Z$ operator shown in Fig.~\ref{fig:Qd-cc-Z-error-string}, we obtain the string representation of the operator $Z_u Z_v$ as shown in Fig.~\ref{fig:Qd-cc-string-pauli-12}.
Note the operator $Z_u Z_v$ commutes with the $X$ stabilizer on blue and green faces.
As a result, these faces do not carry any syndrome and hence the strings do not terminate on these faces.
On the other hand, the operator $Z_u Z_v$ violates the $X$ stabilizer on the red faces and hence red faces carry syndrome $\epsilon_r$ due to violation of $X$ stabilizer, see Fig.~\ref{fig:Qd-cc-string-pauli-2}.
Note that from Fig.~\ref{fig:Qd-cc-X-error-string} and Fig.~\ref{fig:Qd-cc-Z-error-string}, a string terminates on faces with nonzero syndrome.
Therefore, the string has its end points in red faces, see Fig.~\ref{fig:Qd-cc-string-pauli-3}.
To simplify the string representation, we adopt the notation shown in Fig.~\ref{fig:Qd-cc-string-pauli-4} which is equivalent to that in Fig.~\ref{fig:Qd-cc-string-pauli-3}. 

Similar arguments can be used to obtain the string representation of the operator $X_u^\dagger X_v, u \in \mathsf{V}_e$ and $v \in \mathsf{V}_o$, shown in Fig.~\ref{fig:Qd-cc-string-pauli-5}.
Strings for the individual $X$ operators are shown in Fig.~\ref{fig:Qd-cc-string-pauli-13}.
The string representation for the operator $X_u^\dagger X_v$ is given in Fig.~\ref{fig:Qd-cc-string-pauli-7} and its simplified representation of string is shown in Fig.~\ref{fig:Qd-cc-string-pauli-8}. 

\begin{figure}[htb]
    \centering
    \begin{subfigure}{.225\textwidth}
        \centering
        \includegraphics[scale = .85]{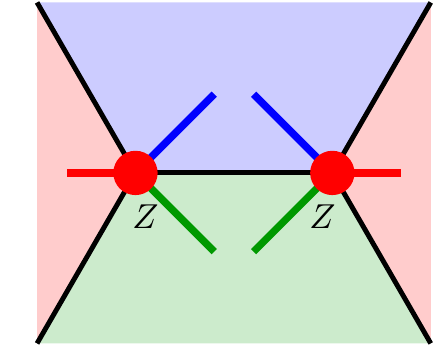}
        \subcaption{}
        \label{fig:Qd-cc-string-pauli-12}
    \end{subfigure}
    ~
    \begin{subfigure}{.225\textwidth}
        \centering
        \includegraphics[scale = .85]{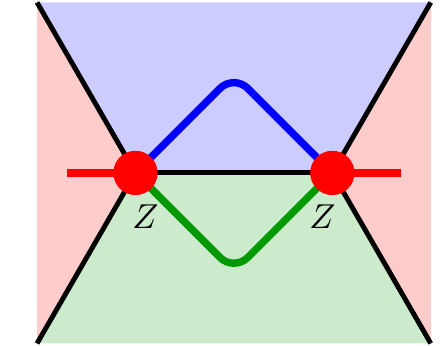}
        \subcaption{}
        \label{fig:Qd-cc-string-pauli-3}
    \end{subfigure}
    ~
    \begin{subfigure}{.225\textwidth}
        \centering
        \includegraphics[scale = .85]{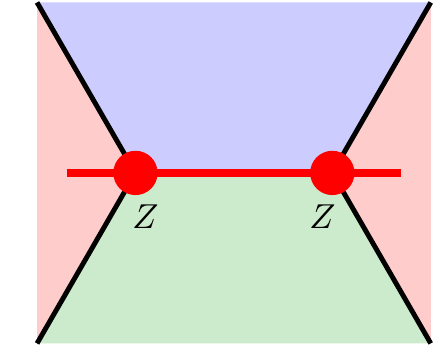}
        \subcaption{}
        \label{fig:Qd-cc-string-pauli-4}
    \end{subfigure}
    ~
    \begin{subfigure}{.225\textwidth}
        \centering
        \includegraphics[scale = .85]{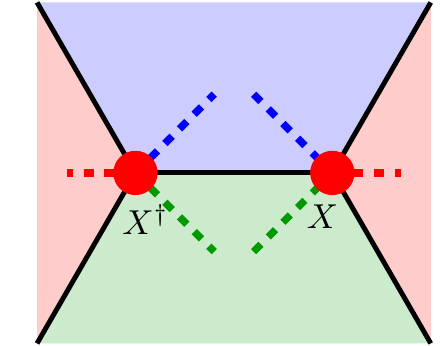}
        \subcaption{}
        \label{fig:Qd-cc-string-pauli-13}
    \end{subfigure}
    ~
    \begin{subfigure}{.225\textwidth}
        \centering
        \includegraphics[scale = .85]{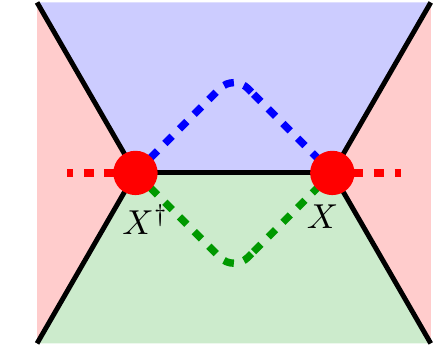}
        \subcaption{}
        \label{fig:Qd-cc-string-pauli-7}
    \end{subfigure}
    ~
    \begin{subfigure}{.225\textwidth}
        \centering
        \includegraphics[scale = .85]{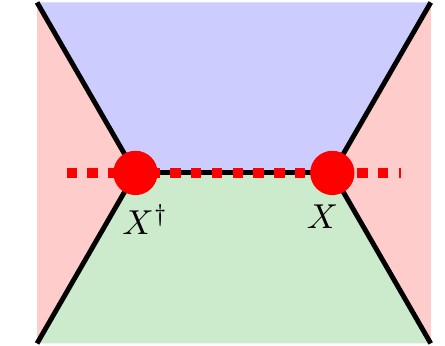}
        \subcaption{}
        \label{fig:Qd-cc-string-pauli-8}
    \end{subfigure}
    \caption{Generalized Pauli operator to string mapping in a 2-colex. (a) Syndrome on blue and green faces vanish as the $X$ stabilizer commutes with the $Z$ error. But red faces host syndromes as stabilizer is violated. (b) The string shown in Fig.~\ref{fig:Qd-cc-string-pauli-3} is represented as shown. Note that this operator can also be seen as the one transporting syndromes from one face to another. (c) Syndromes on green and blue faces vanish whereas that on red faces remain. (d) Blue and green strings in Fig.~\ref{fig:Qd-cc-string-pauli-7} are combined to obtain the string as shown.
    }
    \label{fig:Qd-cc-string-pauli}
\end{figure}

\subsection{String algebra of generalized Pauli operators in the presence of twists}
Recall that the stabilizer generators on twist and modified face are not longer of $Z$ type or  $X$ type.
As a result, the string algebra introduced before needs modification to take into account the modified stabilizers. 
Consider Fig.~\ref{fig:Qd-cc-string-pauli-9}.
The upper red face is twist and the bottom blue face is a modified face.
The blue and green faces on the left and right respectively are unmodified faces.
The syndromes resulting from applying the operator $Z_u X_v$ is shown in Fig.~\ref{fig:Qd-cc-string-pauli-9}.
Note that the syndrome to the left of the domain wall is of green color and not blue.
The reason is that this syndrome can be moved to a green face without crossing the domain wall.
Using the string representation given in  Fig.~\ref{fig:Qd-cc-error-strings}, we arrive at the string representation for the individual operators as shown in Fig.~\ref{fig:Qd-cc-string-pauli-14}.
Note that the error operator commutes with the stabilizers defined on modified face and twist and anticommutes with the stabilizers on blue and green unmodified faces.
Therefore, strings are continuous in twist and modified face, and terminate in the blue and green unmodified faces, see Fig.~\ref{fig:Qd-cc-string-pauli-10}.
The simplified version of strings in Fig.~\ref{fig:Qd-cc-string-pauli-10} is shown in Fig.~\ref{fig:Qd-cc-string-pauli-11}.
Note that the string changes both color and charge as it crosses the domain wall (indicated as dashed line).

\begin{figure}[htb]
    \centering
       \begin{subfigure}{.225\textwidth}
        \centering
        \includegraphics[scale = .85]{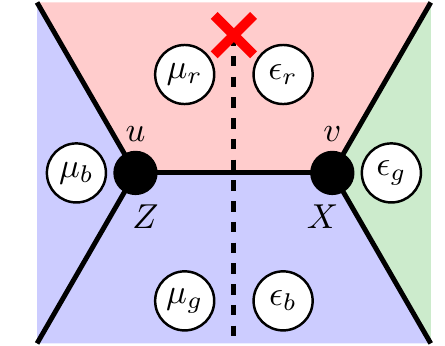}
        \subcaption{}
        \label{fig:Qd-cc-string-pauli-9}
    \end{subfigure}
    ~
    \begin{subfigure}{.225\textwidth}
        \centering
        \includegraphics[scale = .85]{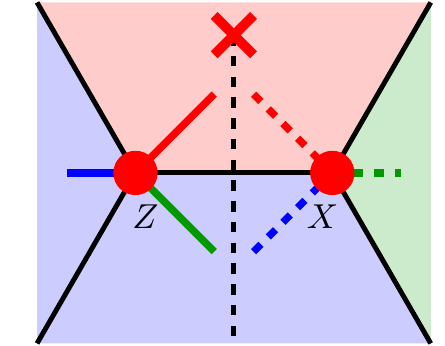}
        \subcaption{}
        \label{fig:Qd-cc-string-pauli-14}
    \end{subfigure}
    ~
    \begin{subfigure}{.225\textwidth}
        \centering
        \includegraphics[scale = .85]{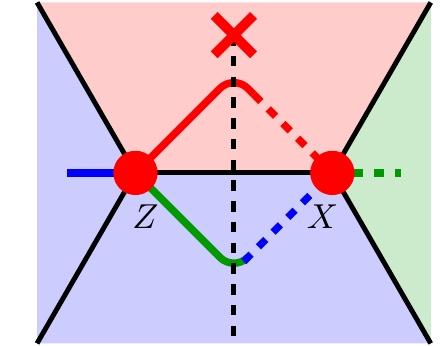}
        \subcaption{}
        \label{fig:Qd-cc-string-pauli-10}
    \end{subfigure}
    ~
    \begin{subfigure}{.225\textwidth}
        \centering
        \includegraphics[scale = .85]{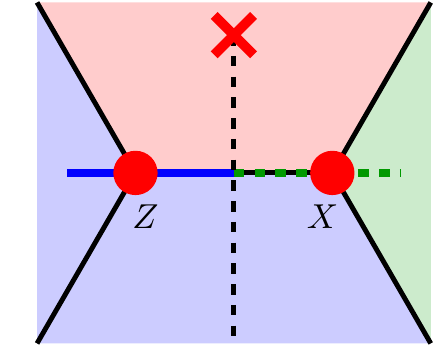}
        \subcaption{}
        \label{fig:Qd-cc-string-pauli-11}
    \end{subfigure}
    \caption{Generalized Pauli operator to string mapping in a lattice with charge-and-color-permuting twists. (a) The operators $Z$ and $X$ create the syndromes as shown. Upper red face is twist and the middle blue face is modified face. (b) The operators $Z$ and $X$ as shown in Fig.~\ref{fig:Qd-cc-string-pauli-9} commute with the stabilizers of twist and modified face. Hence, strings are continuous in these faces. (c) The string in Fig.~\ref{fig:Qd-cc-string-pauli-10} is represented as above for simplicity. Note that the string changes both color and the operator it represents as it crosses the domain wall.}
    \label{fig:Qd-cc-string-pauli-twists}
\end{figure}

We use the string notation developed here to represent stabilizers of twist and modified faces.
Stabilizers of modified face are shown in Fig.~\ref{fig:Qd-cc-charge-color-perm-modified-string-1}~--~\ref{fig:Qd-cc-charge-color-perm-modified-string-4}.  
The strings change both color and charge as they crosses the domain wall.
The string representation of twist stabilizer is shown in Fig.~\ref{fig:Qd-cc-charge-color-perm-twist}.
Here, the string crosses the domain wall twice.
Note that we have given the string representation only for a subset of stabilizer generators.
Our assignment of stabilizer generators is such that it involves only $X$, $X^\dagger$, $Z$ and $Z^\dagger$.

\begin{figure}[htb]
    \centering
    \begin{subfigure}{.225\textwidth}
        \centering
        \includegraphics[scale = .85]{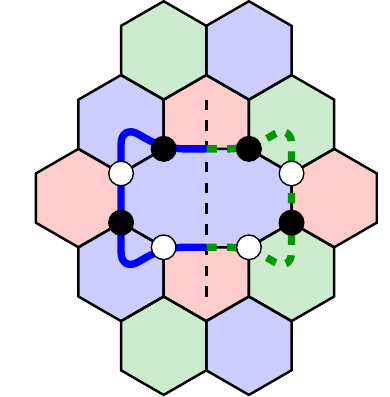}
        \subcaption{}
        \label{fig:Qd-cc-charge-color-perm-modified-string-1}
    \end{subfigure}
    ~
    \begin{subfigure}{.225\textwidth}
        \centering
        \includegraphics[scale = .85]{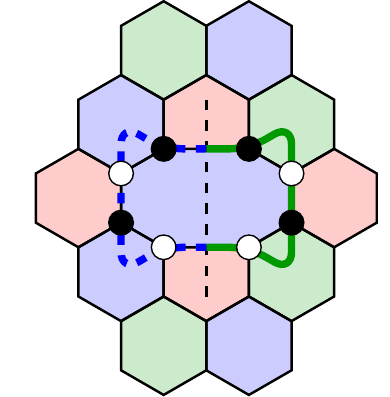}
        \subcaption{}
        \label{fig:Qd-cc-charge-color-perm-modified-string-2}
    \end{subfigure}
        ~
    \begin{subfigure}{.225\textwidth}
        \centering
        \includegraphics[scale = .7]{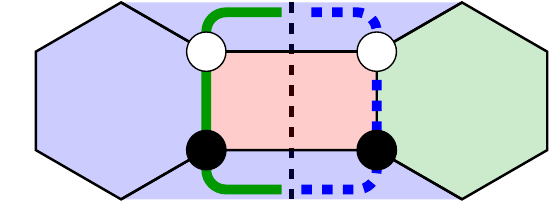}
        \subcaption{}
        \label{fig:Qd-cc-charge-color-perm-modified-string-3}
    \end{subfigure}
        ~
    \begin{subfigure}{.225\textwidth}
        \centering
        \includegraphics[scale = .7]{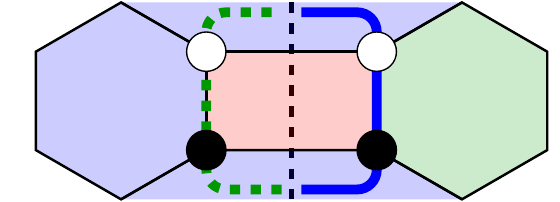}
        \subcaption{}
        \label{fig:Qd-cc-charge-color-perm-modified-string-4}
    \end{subfigure}
    ~
    \begin{subfigure}{.225\textwidth}
        \centering
        \includegraphics[scale = .8]{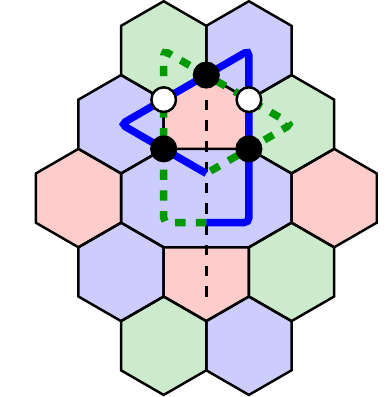}
        \subcaption{}
        \label{fig:Qd-cc-charge-color-perm-twist}
    \end{subfigure}
    \caption{String representation of stabilizers. (a)  This stabilizer has $Z$ operators to the left and $X$ operators to the right of the domain wall. (b) The second stabilizer on modified in string notation which is a mirror image of the first stabilizer. (c)-(d) String representation of stabilizers defined on a red modified face. (e) Twist stabilizer represented in string notation. Note that the string crosses the domain wall twice. }
    \label{fig:Qd-cc-charge-color-perm-modified-twists-string}
\end{figure}

\subsection{Logical Operators}
We now present the logical operators for the encoded qudits.
Logical operators, by definition, commute with all the stabilizer generators and  therefore they do not produce any syndrome. 
Intuitively, we expect that the  string representation of logical operators must be not have any termination, 
in other words, they must be closed strings. 
Logical operators correspond to strings that encircle a pair of twists and are not generated by stabilizers.
Also, logical operators commute with all stabilizers.
This brings forth the question of assigning weights to vertices in the support of logical operators.
Commutation with the stabilizer generators is considered while assigning weights to qudits in the support of logical operators.
The weight assignment is done by taking the stabilizer weight into account.

\begin{figure}[htb]
    \centering
    \begin{subfigure}{.45\textwidth}
        \centering
        \includegraphics[height = 4.75cm, width = 6cm]{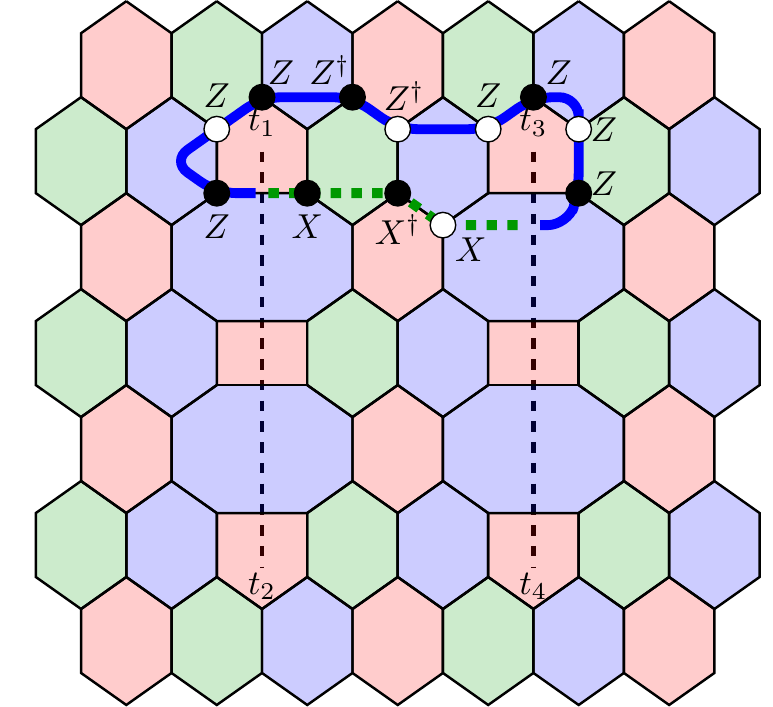}
        \subcaption{}
        \label{fig:Qd-cc-LO-X-charge-color-permuting-lattice}
    \end{subfigure}
    ~
    \begin{subfigure}{.45\textwidth}
        \centering
        \includegraphics[height = 4.75cm, width = 6cm]{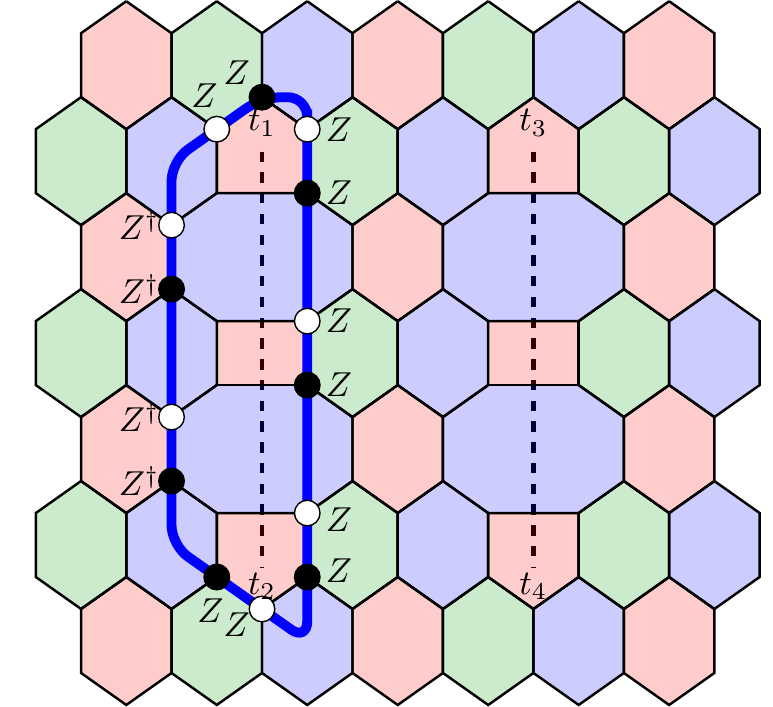}
        \subcaption{}
        \label{fig:Qd-cc-LO-Z-charge-color-permuting-lattice}
    \end{subfigure}
    \caption{Logical operators depicted on lattice.  (a) Logical $X$ operator depicted on the lattice. It encircles twists not created together. Note that the color and operator change as the string crosses the domain wall. (b) Logical $Z$ operator depicted on lattice. It encircles twists created together.}
    \label{fig:Qd-cc-LO-charge-color-permuting-lattice}
\end{figure}

The assignment rules are different for $Z$ and $X$ logical operators.
Recall that the vertices in the qudit color code are partitioned into even and odd vertices.
Let the vertices $u$ and $v$ of a face $f$ be in support of the logical operator.
{We assume that the associated logical operator corresponds to a string that enters through $u$
and exits through $v$.}

\noindent  \emph{Assigning vertex weights for $X$ logical operator.} Note that the $Z$ stabilizer assignment is independent of vertex type and the $X$ operators assigned to vertices $u$ and $v$ of a face $f$ should commute with the $Z$ stabilizer.
Commutation with stabilizers is achieved if the weights associated with $u$ and $v$ differ by $d$, see Fig.~\ref{fig:Qd-cc-LO-X-charge-color-permuting-lattice}.
Therefore, weights for $X$ operator assignment is always $w_{u,x} + w_{v,x} = 0 \mod{d}$  which results in the operator $X_u X_v^\dagger$.

When the string crosses the domain wall, the charge and color change and consequently the rules also change.
Now we have to assign $Z$ and $X$ weights denoted by $w_{u,z}$ and $w_{v,x}$.
Suppose that $u,v \in \mathsf{V}_e$ or  $u \in \mathsf{V}_e, v \in \mathsf{V}_o $, then $w_{u,z} = w_{v,x}$.
If $u,v \in \mathsf{V}_o$ or $u \in \mathsf{V}_o, v \in \mathsf{V}_e$, then $w_{u,z} + w_{v,x} = 0 \mod{d}$.

\noindent \emph{Assigning vertex weights for $Z$ logical operator.} Recall that the $X$ stabilizer assignment is dependent on the vertex type.
If the common vertices $u$ and $v$ between a face $f$ and logical operator are of different kind, i.e. $u \in \mathsf{V}_e$ and $v \in \mathsf{V}_o$ or vice versa, then the logical operator weight associated to the vertices is the same.
The reason being that the restriction of the $X$ stabilizer of the face to these vertices is $X_u X_v^\dagger$ and if the restriction of logical operator to these vertices is  $Z_u(w_{u,z}) Z_v(w_{v,z})$, then commutation with the face stabilizer forces the constraint $w_{u,z} = w_{v,z}$.
If $u$ and $v$ are of the same kind i.e. $u,v \in \mathsf{V}_e$ or $u,v \in \mathsf{V}_o$, the logical operators weights associated to the vertices differ by $d$ i.e. $w_{u,z} + w_{v,z} = 0 \mod{d}$.

The rules for assigning weights to $X$ and $Z$ logical operators are summarized as follows:
\textit{if two vertices $u,v \in V(f)$ are in the support of a logical operator, then,} 
\begin{compactenum}[(i)]
\item \emph{$X$ logical operator}: 
\begin{compactenum}
\item if string does not cross domain wall, then $w_{u,x} + w_{v,x} = 0 \mod{d}$.
\item if string crosses domain wall, then
\begin{compactenum}
\item if $u,v \in \mathsf{V}_e$ or  $u \in \mathsf{V}_e, v \in \mathsf{V}_o $, then $w_{u,z} = w_{v,x}$.
\item if $u,v \in \mathsf{V}_o$ or $u \in \mathsf{V}_o, v \in \mathsf{V}_e$, then $w_{u,z} + w_{v,x} = 0 \mod{d}$.
\end{compactenum}
\end{compactenum}

\item \emph{$Z$ logical operator}:
\begin{compactenum}
\item if $u, v \in \mathsf{V}_e (\text{ or } \mathsf{V}_o)$, then $w_{u,z} = w_{v,z}$,
\item else, $w_{u,z} + w_{v,z} = 0 \mod{d}$. 
\end{compactenum}
\end{compactenum}
\vspace{2mm}

A set of logical operators when four twists are present in the lattice is shown in Fig.~\ref{fig:Qd-cc-LO-charge-color-permuting-lattice}.
The operators shown commute with all stabilizer generators.
These two operators do not commute as the generalized Pauli operator is different on one of the common vertices to the operators. 
Hence, these are not generated by stabilizers and must be nontrivial logical operators.
Observe that logical operators are generated by $X$, $X^\dagger$, $Z$ and $Z^\dagger$ like the stabilizers.
This allows us to represent them  as strings on the lattice without explicitly mentioning the weights of each error.

The full set of logical operators when six twists are present in the lattice is shown in Fig.~\ref{fig:Qd-cc-LO-charge-color}.
Since we have abstracted out the lattice, there is no way to indicate the powers of the generalized Pauli operators in the support of a logical operator.
Therefore, we assign a direction to indicate the choice of powers of generalized Pauli operators in the support of logical operators.
We use clockwise direction to represent the operators shown in Fig.~\ref{fig:Qd-cc-LO-charge-color-permuting-lattice}.
The inverse of these operators will be indicated as strings with counterclockwise direction.

\begin{figure}[htb]
    \centering
    \begin{subfigure}{.45\textwidth}
        \centering
        \includegraphics[scale = .85]{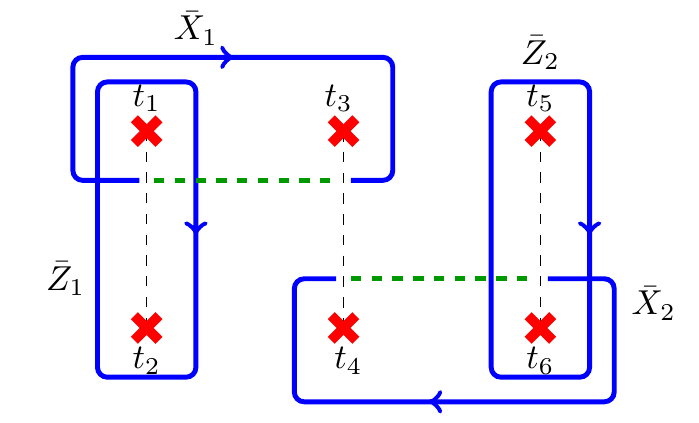}
        \subcaption{}
        \label{fig:Qd-cc-LO-charge-color-canonical}
    \end{subfigure}
    ~
    \begin{subfigure}{.45\textwidth}
        \centering
        \includegraphics[scale = .85]{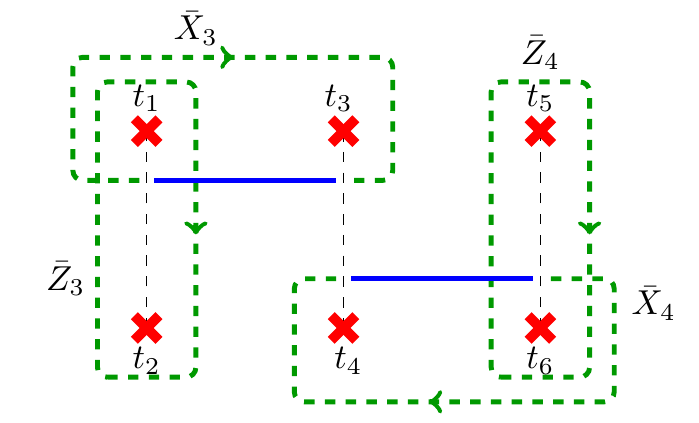}
        \subcaption{}
        \label{fig:Qd-cc-LO-charge-color-full}
    \end{subfigure}
    ~
    \begin{subfigure}{.45\textwidth}
        \centering
        \includegraphics[scale = .85]{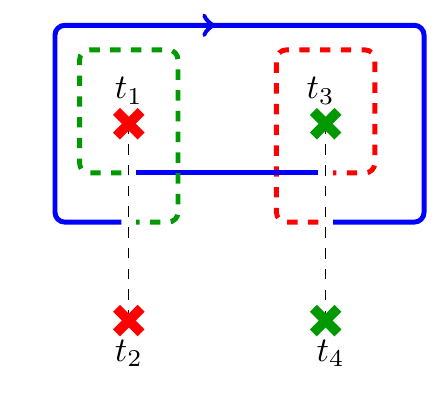}
        \subcaption{}
        \label{fig:Qd-cc-self-intersecting-LO}
    \end{subfigure}
    
    \caption{Full set of logical operators and the canonical logical operators when six twists are present in the lattice, see Fig.~\ref{fig:Qd-cc-error-strings} for the string notation used. (a) The canonical form logical operators when six twists are present in the lattice. (b) The logical operators that are treated as gauge operators. (c) Self-intersecting logical operator when twists of different color are present.}
    \label{fig:Qd-cc-LO-charge-color}
\end{figure}

When twists of different color are used for encoding, then some of the logical operators will be self intersecting, see Fig.~\ref{fig:Qd-cc-self-intersecting-LO}.
To avoid self-intersecting logical operators, we take all twists to be of red color~\cite{GowdaSarvepalli2021}.
However, this choice of logical operators has the drawback that whenever an single qudit encoded gate is to be performed only on say, encoded qudit 1 (2), we inevitably also perform a gate on encoded qudit 3 (4). 
This is because the logical operators of qubits 1 and 3 encircle the same pair of twists, see Fig.~\ref{fig:Qd-cc-LO-charge-color}.
Suppose that twists $t_1$ and $t_2$ are braided to implement a gate on qudit 1.
Because of braiding twists, a gate is also inevitably implemented on qudit 3.
This is undesirable.
Therefore, we treat encoded qubits 3 and 4 in Fig.~\ref{fig:Qd-cc-LO-charge-color-full} as gauge qudits.
The canonical form of logical operators we use are shown in Fig.~\ref{fig:Qd-cc-LO-charge-color-canonical}.
With the choice of logical operators given in Fig.~\ref{fig:Qd-cc-LO-charge-color-canonical}, we get $\lfloor t / 3  \rfloor $ encoded qudits when $t$ twists are present in the lattice.
The rest of the $t - 2 - \lfloor t / 3  \rfloor $ logical qudits are treated as gauge qudits.

\begin{theorem}[Construction]
A qudit color code with $t$ charge-and-color-permuting twists with $\lfloor t / 3\rfloor$ encoding defines a subsystem code with $\lfloor t / 3\rfloor$ logical qudits and $t - 2 - \lfloor t / 3  \rfloor$ gauge qudits.
\label{thm:construction-summary}
\end{theorem}
Note that the subsystem codes defined here are not the same as the topological subsystem codes in \cite{Bombin2011} which are defined using two body gauge operators. 
Using the construction of charge-and-color-permuting twists and the string formalism for generalized Pauli operators, we next proceed to present protocols to implement encoded generalized Clifford gates using charge-and-color-permuting twists. 

\section{Generalized Clifford Gates with charge-and-color-permuting twists}
\label{sec:gates}
In this section, we discuss the implementation of encoded generalized Clifford gates using charge-and-color-permuting twists.
The generalized Clifford gates are given below~\cite{Grassl2003}:
\begin{subequations}
\begin{eqnarray}
M_{\gamma} &=& \sum_{x \in \mathbb{F}_{d}} |\gamma x \rangle \langle x |, \gamma \ne 0,\\
P_{\gamma} &=& \sum_{x \in \mathbb{F}_{d}} \omega^{\gamma x^2  / 2} | x \rangle \langle x |, \label{eq:subeq-phase}\\
F &=& \frac{1}{\sqrt{d}} \sum_{x,y \in \mathbb{F}_{d}} \omega^{xy} |y\rangle \langle x|, \\
\text{CNOT}(i,j) &=& \sum_{x,y \in \mathbb{F}_{d}} |x\rangle_i |x + y\rangle_{j} \langle x|_i \langle y|_j.
\end{eqnarray}
\label{eqn:gen-clifford-gates}
\end{subequations}

We now present the protocols to implement generalized Clifford gates by braiding and Pauli frame update.
For single qubit gates, we use Pauli frame update and braiding and for entangling gate, we use braiding.
The encoding used is that given in Fig.~\ref{fig:Qd-cc-LO-charge-color-canonical}.
To prove the correctness of the proposed encoded gates,  we need to verify that the logical operators transform 
as required under conjugation.
\medskip

\noindent \emph{Multiplier gate.}
The conjugation relation for multiplier gate $M_\gamma$ is given below:
\begin{subequations}
\begin{eqnarray}
M_{\gamma} X(\alpha) M_{\gamma}^{-1} &=& X(\gamma \alpha), \\
M_{\gamma} Z(\beta) M_{\gamma}^{-1} &=& Z(\gamma^{-1} \beta).
\end{eqnarray}
\label{eqn:multiplier-gate}
\end{subequations}

We use Pauli frame update~\cite{Hastings2015, GowdaSarvepalli2020} to implement this gate.
Pauli frame update is done classically where the Pauli frame (the information regarding the interchange of Pauli labels of the canonical logical operators) of each encoded qudit is kept track of during computation.
Weights of the operators on vertices in the support of logical operators are updated according to Equations~\eqref{eqn:multiplier-gate}.
For a vertex in the support of $\Bar{X}$ with weight $w$, the updated weight is $\gamma w$ and for that in the support of $\Bar{Z}$, the updated weight is $\gamma^{-1} w$.
\medskip

\noindent \emph{DFT gate.}
The DFT gate is the generalization of the Hadamard gate in the case of qubits.
The conjugation relation for the DFT gate $F$ is given below.
\begin{subequations}
\begin{eqnarray}
F X(\alpha) F^{-1} &=& Z(\alpha) \\
F Z(\alpha) F^{-1} &=& X(-\alpha)
\end{eqnarray}
\end{subequations}
The DFT gate is implemented by Pauli frame update.
The Pauli frame update required is $X(\alpha) \rightarrow Z(\alpha)$, $Z(\alpha) \rightarrow X(-\alpha)$.
This is interchanging the labels of $Z$ and $X$ logical operators and reversing the direction of the new logical $Z$ operator.
\medskip

\noindent \emph{Phase gate.}
The Phase gate $P_\gamma$ has the following conjugation relation.
\begin{subequations}
\begin{eqnarray}
P_{\gamma} X(\alpha) P_{\gamma}^{-1} &=& \omega^{-\gamma \alpha^2 / 2} X(\alpha) Z(-\gamma \alpha) \\
P_{\gamma} Z(\alpha) P_{\gamma}^{-1} &=& Z(\alpha)
\end{eqnarray}
\label{eqn:phase-gate}
\end{subequations}

To realize the gate $P_{1}$, we braid twists $t_1$ and $t_2$ counterclockwise as shown in Fig.~\ref{fig:Qd-cc-phase-gate-1}.
The operator $\Bar{Z}$ encircles twists $t_1$ and $t_2$ and hence is left unchanged by braiding.
The deformation of logical $X$ operator and its equivalence (up to a gauge) to the operator $XZ^\dagger$ is shown in Appendix~\ref{sec:XZdagger}. 
From Eqn.~\eqref{eq:subeq-phase}, it can be inferred that $ P_{\gamma} = P_{1}^\gamma$.
The gate $P_\gamma$ is accomplished by performing the $P_1$ gate $\gamma$ number of times.
Note that the gate $P_{d-1}$ is the conjugate of the gate $P_1$. 
Hence it can be realized by braiding twists $t_1$ and $t_2$ clockwise as shown in Fig.~\ref{fig:Qd-cc-phase-gate-2}.

\begin{figure}[htb]
    \centering
    \begin{subfigure}{.225\textwidth}
        \centering
        \includegraphics[scale = .85]{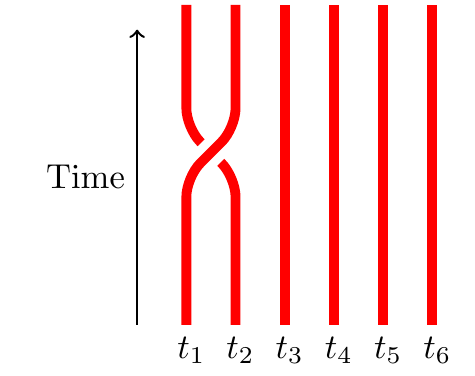}
        \subcaption{}
        \label{fig:Qd-cc-phase-gate-1}
    \end{subfigure}
    ~
    \begin{subfigure}{.225\textwidth}
        \centering
        \includegraphics[scale = .85]{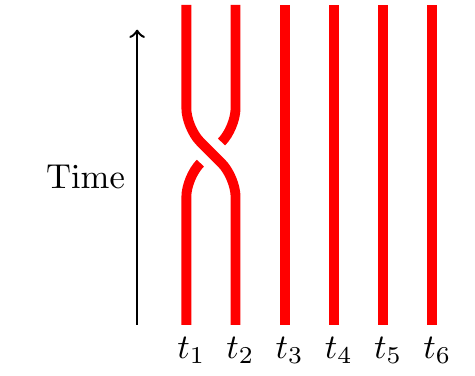}
        \subcaption{}
        \label{fig:Qd-cc-phase-gate-2}
    \end{subfigure}
    \caption{ (a) Braiding twists $t_1$ and $t_2$ counterclockwise to realize the gate $P_1$. (b) Twists $t_1$ and $t_2$ are braided clockwise to realize the gate $P_{d-1}$ ($P^\dagger$).
    }
    \label{fig:Qd-cc-phase-gate}
\end{figure}
\medskip

\noindent CNOT \emph{gate.}
We implement the CNOT gate by using controlled-$Z^\dagger$ ($CZ^\dagger$) gate and DFT gates.
The description of $CZ^\dagger$ gate is given below~\cite{GowdaSarvepalli2020}:

\begin{equation}
   CZ^\dagger(a,b)  = \sum_{x \in \mathbb{F}_d} |x\rangle \langle x|_a \otimes {Z_b^\dagger}(x).
\end{equation}
The protocol for realizing controlled-$Z^\dagger$ gate is given below and is shown in Fig.~\ref{fig:Qd-cc-CZ-dagger-1}.

\medskip

\begin{tabular}{ll}
\hline
 & Controlled-$Z^\dagger$ gate protocol \\
\hline
   1.  &  Braid twists $t_3$ and $t_4$ counterclockwise. \\
    2. &  Perform $P^\dagger$ gate on control and target qubits.\\
\hline
\end{tabular}
\medskip
\begin{figure}[htb]
    \centering
    \begin{subfigure}{.225\textwidth}
        \centering
        \includegraphics[scale = .825]{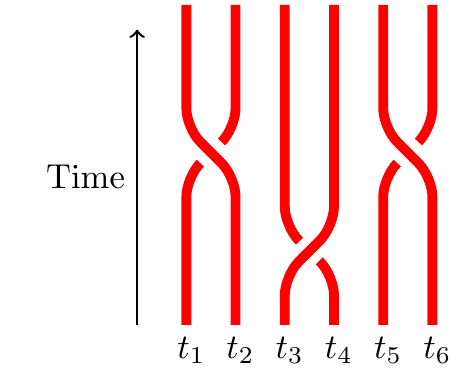}
        \subcaption{}
        \label{fig:Qd-cc-CZ-dagger-1}
    \end{subfigure}
     ~
    \begin{subfigure}{.45\textwidth}
        \centering
        \includegraphics[scale = .85]{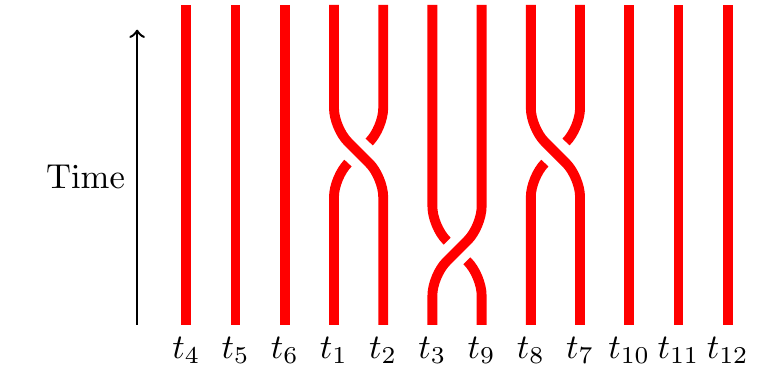}
        \subcaption{}
        \label{fig:Qd-cc-CZ-dagger-4}
    \end{subfigure}
    \caption{Realizing controlled-$Z^\dagger$ on qudits in the same block and qudits from adjacent blocks. (a) Controlled-$Z^\dagger$ gate on logical qudits in the same block is achieved by braiding twists $t_3$ and $t_4$ counterclockwise and twists $t_1$, $t_2$ and $t_5$, $t_6$ clockwise. (b) For doing controlled-$Z^\dagger$ gate between qubits of different block, braid the twists as shown. Note that the twists are not labeled sequentially for the sake of clarity.}
    \label{fig:Qd-cc-CZ-dagger}
\end{figure}

The evolution of logical operators is given below.
\begin{eqnarray*}
\begin{array}{ccccc}
     \Bar{Z}_1 & \stackrel{{(1)}}{\longrightarrow} & \Bar{Z}_1  & \stackrel{{(2)}}{\longrightarrow} & \Bar{Z}_1  \\
   \Bar{X}_1 & \stackrel{{(1)}}{\longrightarrow} & \Bar{X}_1 \Bar{Z}_1^\dagger \Bar{Z}_2^\dagger  & \stackrel{{(2)}}{\longrightarrow} & \Bar{X}_1  \Bar{Z}_2^\dagger  \\
   \Bar{Z}_2 & \stackrel{{(1)}}{\longrightarrow} & \Bar{Z}_2  & \stackrel{{(2)}}{\longrightarrow} & \Bar{Z}_2 \\
   \Bar{X}_2 & \stackrel{{(1)}}{\longrightarrow} & \Bar{Z}_1^\dagger \Bar{Z}_2^\dagger \Bar{X}_2  & \stackrel{{(2)}}{\longrightarrow} &   \Bar{Z}_1^\dagger \Bar{X}_2
\end{array}
\end{eqnarray*}

Logical $Z$ operators of both qudits are left unchanged by braiding twists $t_3$ and $t_4$. 
However, logical $X$ operator is deformed as shown in Fig.~\ref{fig:Qd-cc-CZ-dagger-3}.
The deformed string can be expressed as a combination of string encircling twists  $t_3$ and $t_4$ and the logical $X$ operator.
Note that the string encircling twists $t_3$ and $t_4$ is $\Bar{Z}_1^\dagger \Bar{Z}_2^\dagger$, see Fig.~\ref{fig:Qd-cc-CZ-dagger-2}.
Hence, the string shown in Fig.~\ref{fig:Qd-cc-CZ-dagger-3} corresponds to the operator $\Bar{X}_1 \Bar{Z}_1^\dagger \Bar{Z}_2^\dagger$.
By symmetry, it can be argued that $\Bar{X}_2$ is mapped to $\Bar{Z}_1^\dagger \Bar{Z}_2^\dagger \Bar{X}_2$.
After braiding twists $t_3$ and $t_4$, performing $P^\dagger$ gate on both qudits will lead to $CZ^\dagger$ gate between them.

\begin{figure}[htb]
    \centering
    \begin{subfigure}{.45\textwidth}
        \centering
        \includegraphics[scale = 1]{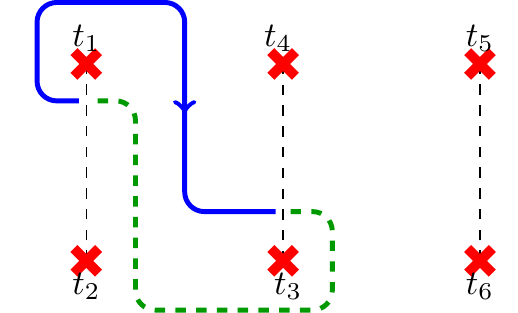}
        \subcaption{}
        \label{fig:Qd-cc-CZ-dagger-3}
    \end{subfigure}
    ~
    \begin{subfigure}{.45\textwidth}
        \centering
        \includegraphics[scale = .9]{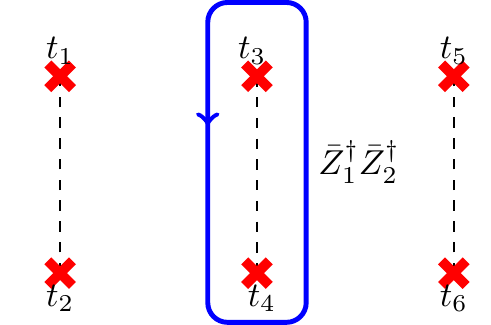}
        \subcaption{}
        \label{fig:Qd-cc-CZ-dagger-2}
    \end{subfigure}
    \caption{ (a) The logical $X_1$ operator is deformed as shown. This operator is the combination of $\Bar{Z}_1^\dagger \Bar{Z}_2^\dagger$ and $X_1$. (b) The blue string encircling twists $t_3$ and $t_4$ is $\Bar{Z}_1^\dagger \Bar{Z}_2^\dagger$. }
    \label{fig:Qd-cc-CZ-dagger-proof}
\end{figure}

So far, we considered qudits from one block which is a set of three pairs of twists that encode two logical qudits.
We also need to be able to perform the CNOT gate between qudits from different blocks.
The same procedure works for qudits from different blocks.
Suppose that twists $t_1$, $t_2$ and $t_3$ encode qudit 1 and twists $t_7$, $t_8$ and $t_9$ encode qudit 3, see Fig.~\ref{fig:Qd-cc-CZ-dagger-4}.
The string encircling twists $t_3$ and $t_9$ can be shown to be $Z_1^\dagger Z_3^\dagger$ in a way similar to that in Ref.~\cite{GowdaSarvepalli2020}.
Braiding twists $t_3$ and $t_9$ followed by $P^\dagger$ gate on control and target qudits will lead to $CZ^\dagger$ gate between them.
The CNOT gate can be obtained from the $CZ^\dagger$ gate by conjugating the target qudit with DFT gate:
\begin{eqnarray*}
F_b (CZ^\dagger (a,b))  F^\dagger_b &=& \sum_{x \in \mathbb{F}_d} |x\rangle \langle x|_a \otimes F_b {Z_b^\dagger}(x) F^\dagger_b \\
  &=& \sum_{x \in \mathbb{F}_d} |x\rangle \langle x|_a \otimes X_b(x)
\end{eqnarray*}
This concludes our discussion on implementation of generalized Clifford gates using charge-and-color-permuting twists.
We next address the issue of lattice modification during braiding.
\medskip

\noindent \emph{Lattice modification during braiding.} Change in the structure of the lattice during twist braiding is shown in Fig.~\ref{fig:Qd-cc-twist-movement}.
In Fig.~\ref{fig:Qd-cc-twist-movement-1}, the initial configuration of twists is shown.
Also shown are the physical qudits that were disentangled during twist creation and movement.
Twist $t_1$ is moved along the path shown in blue in Fig.~\ref{fig:Qd-cc-twist-movement-1} and twist $t_2$ is moved through the modified faces.
There are two distinct steps in braiding: twist movement to an unmodified face as shown in Fig.~\ref{fig:Qd-cc-twist-movement-2} and moving twist to a modified face as shown in Fig.~\ref{fig:Qd-cc-twist-movement-3}.
The latter step involves entangling the physical qudits as shown in Fig.~\ref{fig:Qd-cc-twist-movement-3}.
The portion of the lattice from which twist is moved regains local three colorability, see Fig.~\ref{fig:Qd-cc-twist-movement-4}.
This movement also creates a large modified face. 
However, subsequent twist movement will shrink the size of this face, see Fig.~\ref{fig:Qd-cc-twist-movement-5}.
The subsequent twist movement is shown in Fig.~\ref{fig:Qd-cc-twist-movement-6}, Fig.~\ref{fig:Qd-cc-twist-movement-7} and Fig.~\ref{fig:Qd-cc-twist-movement-8}.
At this stage, the braiding is complete.
However, to minimize the number of modified faces between the twists, we entangle the physical qudits shown in Fig.~\ref{fig:Qd-cc-twist-movement-9} and remove some physical qudits shown in red.
Doing so, we obtain the lattice shown in Fig.~\ref{fig:Qd-cc-twist-movement-10}. 
Note that once the braiding is complete, the large faces resulting from twist movement vanish.
Therefore, braiding twists will not alter the lattice structure drastically, at least in the case of hexagon lattice.
\begin{figure*}
    \centering
    \begin{subfigure}{.3\textwidth}
        \centering
        \includegraphics[scale = .85]{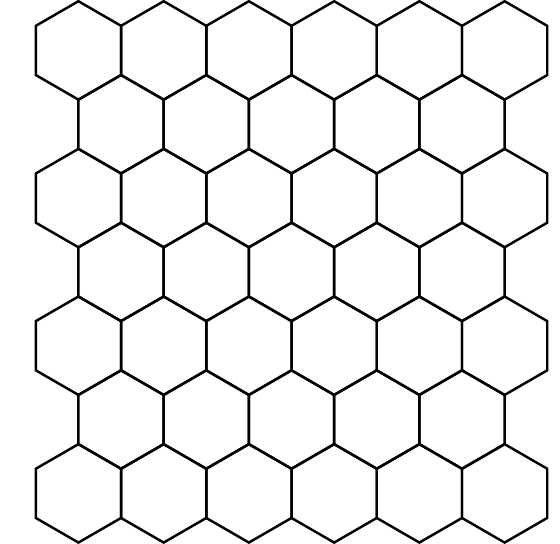}
        \subcaption{}
        \label{fig:Qd-cc-twist-movement-0}
    \end{subfigure}
    ~
    \begin{subfigure}{.3\textwidth}
        \centering
        \includegraphics[scale = .85]{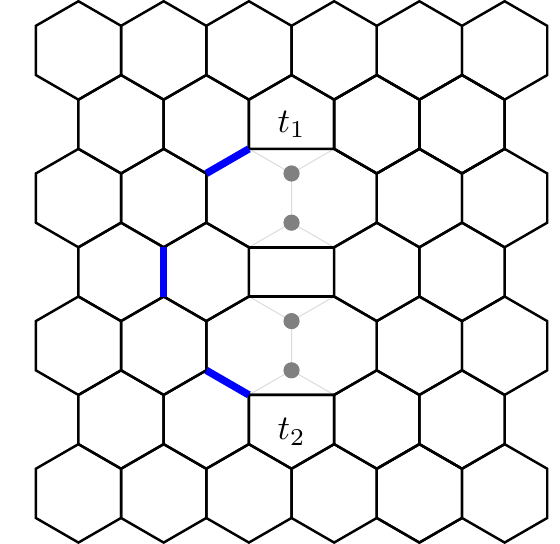}
        \subcaption{}
        \label{fig:Qd-cc-twist-movement-1}
    \end{subfigure}
    ~
    \begin{subfigure}{.3\textwidth}
        \centering
        \includegraphics[scale = .85]{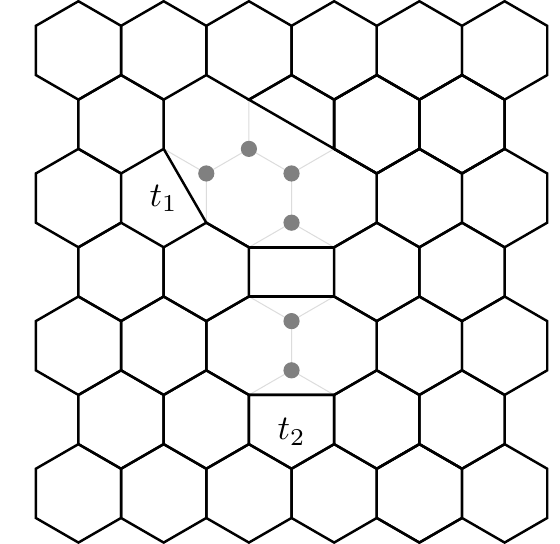}
        \subcaption{}
        \label{fig:Qd-cc-twist-movement-2}
    \end{subfigure}
    ~
    \begin{subfigure}{.3\textwidth}
        \centering
        \includegraphics[scale = .85]{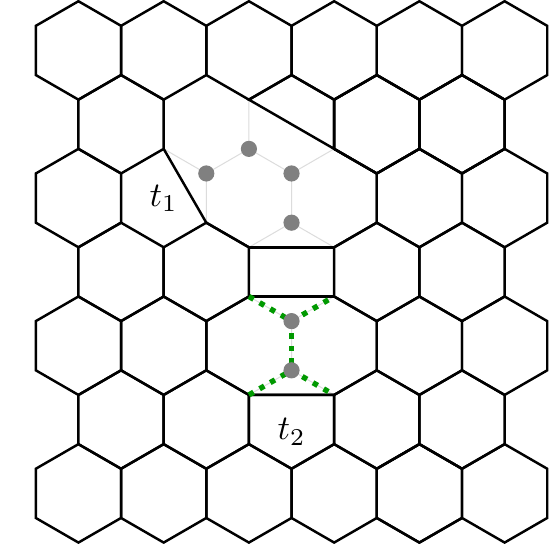}
        \subcaption{}
        \label{fig:Qd-cc-twist-movement-3}
    \end{subfigure}
    ~
    \begin{subfigure}{.3\textwidth}
        \centering
        \includegraphics[scale = .85]{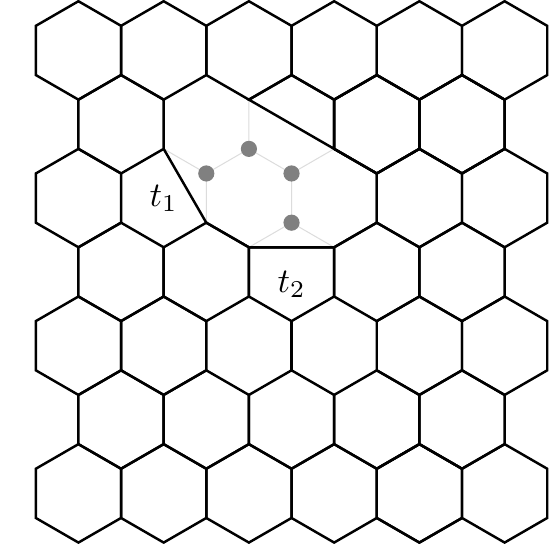}
        \subcaption{}
        \label{fig:Qd-cc-twist-movement-4}
    \end{subfigure}
    ~
    \begin{subfigure}{.3\textwidth}
        \centering
        \includegraphics[scale = .85]{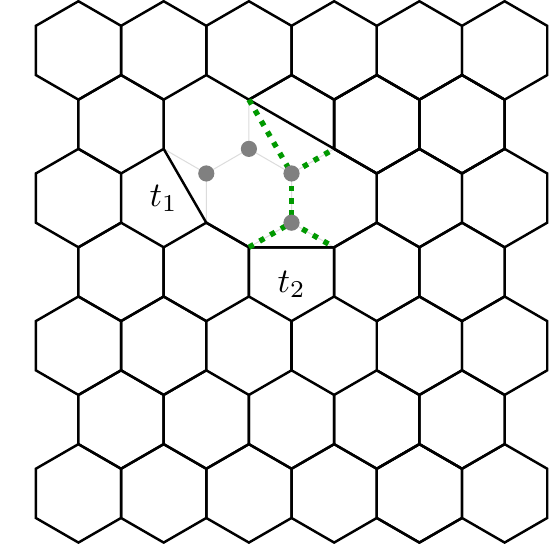}
        \subcaption{}
        \label{fig:Qd-cc-twist-movement-5}
    \end{subfigure}
    ~
    \begin{subfigure}{.3\textwidth}
        \centering
        \includegraphics[scale = .85]{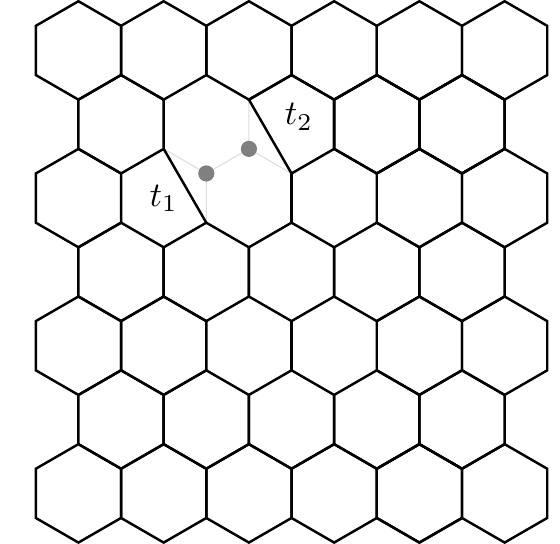}
        \subcaption{}
        \label{fig:Qd-cc-twist-movement-6}
    \end{subfigure}
    ~
    \begin{subfigure}{.3\textwidth}
        \centering
        \includegraphics[scale = .85]{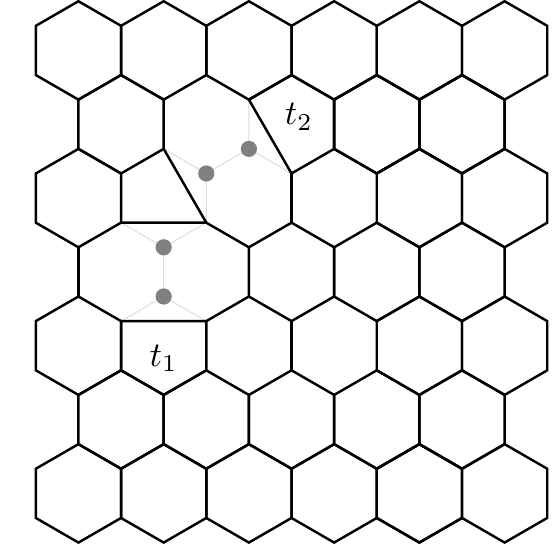}
        \subcaption{}
        \label{fig:Qd-cc-twist-movement-7}
    \end{subfigure}
    ~
    \begin{subfigure}{.3\textwidth}
        \centering
        \includegraphics[scale = .85]{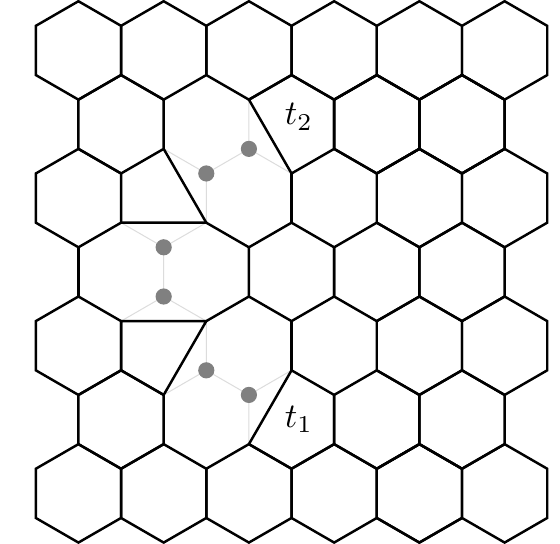}
        \subcaption{}
        \label{fig:Qd-cc-twist-movement-8}
    \end{subfigure}
    ~
    \begin{subfigure}{.3\textwidth}
        \centering
        \includegraphics[scale = .85]{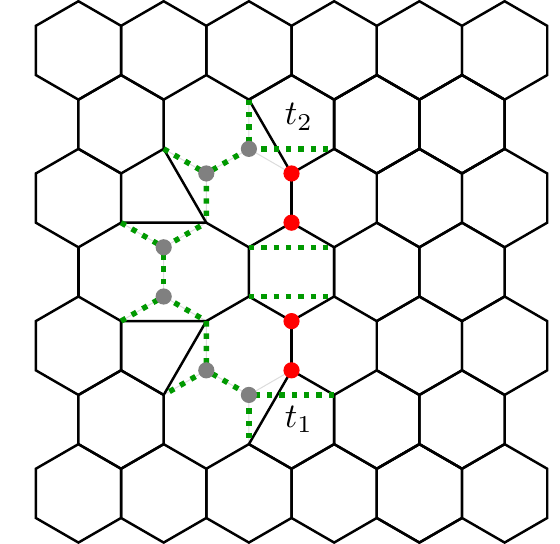}
        \subcaption{}
        \label{fig:Qd-cc-twist-movement-9}
    \end{subfigure}
    ~
    \begin{subfigure}{.3\textwidth}
        \centering
        \includegraphics[scale = .85]{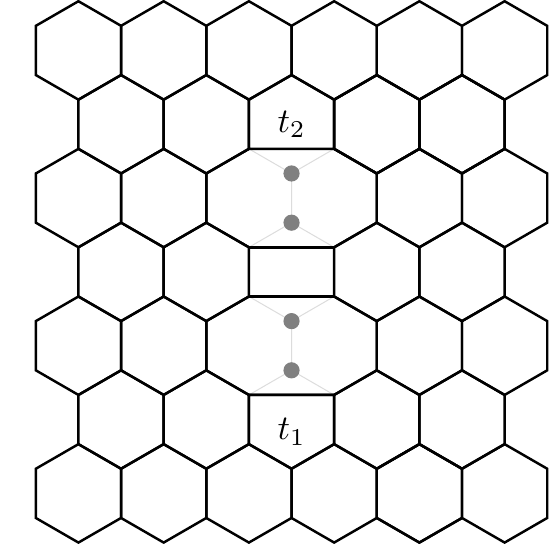}
        \subcaption{}
        \label{fig:Qd-cc-twist-movement-10}
    \end{subfigure}
    \caption{Lattice modification during braiding a pair of twists. The physical qubits disentangled from the code lattice during twist creation and movements are shown in gray color. Dotted lines show the new faces when some of them are added back to the code lattice. 
    Some of the modified faces grow in size during braiding. However, this growth is not permanent. These faces shrink once the braiding process is complete. 
    }
    \label{fig:Qd-cc-twist-movement}
\end{figure*}

\section{Conclusion}
In this paper, we initiated the study of twists in qudit color codes.
Specifically, we studied charge-and-color-permuting twists.
We gave the construction  of qudit color codes with charge-and-color-permuting twists starting from a $2$-colex.
Using the construction, we presented protocols to implement encoded generalized Clifford gates.
A future direction for further investigation could be to explore the set of all twist types possible in qudit color codes.
Since qudit color codes have rich anyon structure, one could expect more types of twist.
Another fruitful direction would be to study the mapping between the proposed codes and qudit surface codes.
Such mappings could lead to efficient decoders for qudit color codes with twists.

\section*{Acknowledgment}
We thank the anonymous reviewers for their comments.
This work was supported by  the Department of Science and Technology, Government of India, under 
Grant No.~DST/ICPS/QuST/Theme-3/2019/Q59.

\section*{Appendix}
\appendix
\section{Proof of Lemma~\ref{lm:stabilizer-commutation}}
\label{sec:stab-commutation}
In this section, we give a proof of stabilizer commutation.
Note that since all vertices in the lattice are trivalent, any two adjacent faces share an edge (and hence two common vertices).
Let $u$, $v$ be the common vertices to two adjacent faces $f$ and $g$.
When vertices $u$ and $v$ come from different bipartition, we take $u \in \mathsf{V}_e$ and $v \in \mathsf{V}_o$ .
We have to consider six cases here:
\begin{compactenum}[i)]
\item $\mathcal{D}_0$ and $\mathcal{D}_0$ : Follows from Equations~\eqref{eqn:qudit-stab}.

\item $\mathcal{D}_0$ and $\mathcal{D}_2$: 
The restriction of both type of stabilizer generators to common vertices is $Z_u Z_v$ and $ X_u X^\dagger_v$.
Since these operators commute, the corresponding stabilizer generators commute.

\item $\mathcal{D}_0$ and $\mathcal{T}$: Reason is similar to the one given above.

\item $\mathcal{D}_2$ and $\mathcal{D}_2$: 
Note that for any two adjacent modified faces, the common vertices belong to the same bipartition.
If the common vertices are in $\mathsf{V}_e$, then the restriction of stabilizer generators to common vertices is $Z_u X_v$ and $X_u Z_v$. 
If the common vertices are in $\mathsf{V}_o$, then $Z_u X^\dagger_v$ and $X^\dagger_u Z_v$ are the restrictions to common vertices.
Commutation of these stabilizer generators follows as the restrictions commute.
Note that the stabilizer generators defined on the same face have exactly the same support.
The phase resulting from exchange of operators on any two adjacent vertices is zero.
Therefore, the total phase resulting from operator exchange is zero and hence the stabilizer generators on the same face commute.

\item $\mathcal{D}_2$ and $\mathcal{T}$: Similar to the case of commutation of two modified faces, a twist and a modified face have common vertices the same bipartition.
In case of common vertices belonging to $\mathsf{V}_e$, the restriction of twist stabilizer is $Z_uX_u Z_vX_v$ and for the modified face we have $Z_u X_v$ and $X_v Z_u$.
For odd vertices, twist stabilizer restriction is $Z_uX^\dagger_u Z_vX^\dagger_v$ and for modified faces we have $Z_u X^\dagger_v$ and $X^\dagger_u Z_v$.
Hence, commutation follows.

\item $\mathcal{T}$ and $\mathcal{T}$: Two twist faces are always separated by a modified face in between. Hence they commute. 
If two twists happen to be adjacent during braiding, then the operators on the common vertices are the same. Therefore, twist stabilizers commute.
\end{compactenum}
Hence, the stabilizers defined in Equations~\eqref{eqn:qudit-stab}, ~\eqref{eqn:qudit-cc-modified-stab},~\eqref{eqn:qudit-cc-twist-stab} commute.

\section{Proof for Phase gate}
\label{sec:XZdagger}

The deformed logical $X$ operator after braiding twists $t_1$ and $t_2$ is shown in Fig.~\ref{fig:Qd-cc-phase-gate-3}.
We now show that this operator is equivalent up to a gauge to the operator $\Bar{X} \Bar{Z}^\dagger$ shown in Fig.~\ref{fig:Qd-cc-phase-gate-4}.

\begin{figure}[htb]
    \centering
    \begin{subfigure}{.225\textwidth}
        \centering
        \includegraphics[scale = .85]{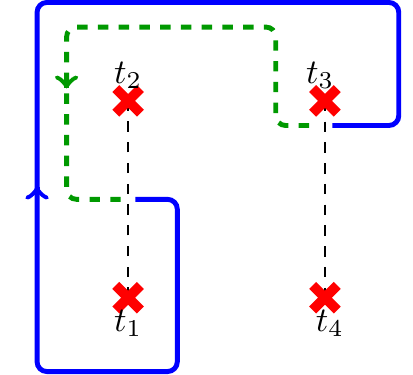}
        \subcaption{}
        \label{fig:Qd-cc-phase-gate-3}
    \end{subfigure}
    ~
    \begin{subfigure}{.225\textwidth}
        \centering
        \includegraphics[scale = .85]{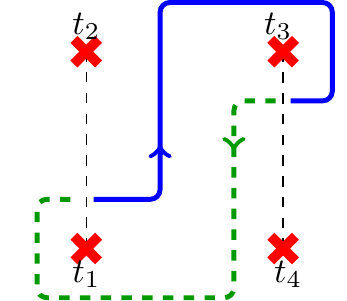}
        \subcaption{}
        \label{fig:Qd-cc-phase-gate-4}
    \end{subfigure}
    \caption{Deformation of logical $X$ operator after braiding twists $t_1$ and $t_2$. (a) The logical $X$ operator is deformed as shown after braiding twists $t_1$ and $t_2$. (b) The deformed logical operator shown in Fig.~\ref{fig:Qd-cc-phase-gate-3} is equivalent up to a gauge to the operator shown.}
    \label{fig:Qd-cc-phase-gate-proof}
\end{figure}

\begin{figure}
\centering
\begin{subfigure}{.225\textwidth}
    \centering
    \includegraphics[scale = .75]{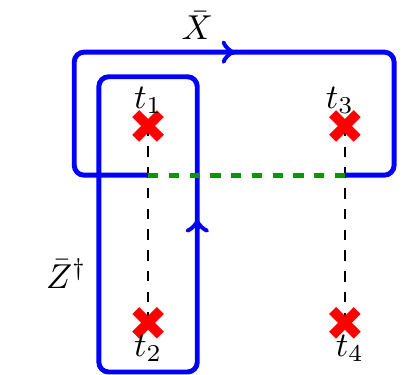}
    \subcaption{}
    \label{fig:Qd-cc-X-Z-dagger-1}
\end{subfigure}
~
\begin{subfigure}{.225\textwidth}
    \centering
    \includegraphics[scale = .75]{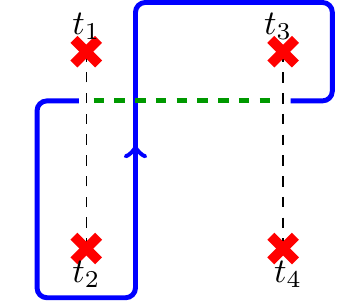}
    \subcaption{} 
    \label{fig:Qd-cc-X-Z-dagger-2}
\end{subfigure}
~
\begin{subfigure}{.225\textwidth}
    \centering
    \includegraphics[scale = .75]{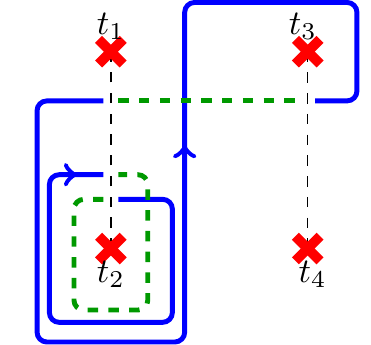}
    \subcaption{}
    \label{fig:Qd-cc-X-Z-dagger-3}
\end{subfigure}
~
\begin{subfigure}{.225\textwidth}
    \centering
    \includegraphics[scale = .75]{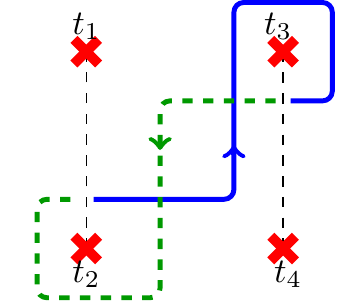}
    \subcaption{}
    \label{fig:Qd-cc-X-Z-dagger-4}
\end{subfigure}
~
\begin{subfigure}{.225\textwidth}
    \centering
    \includegraphics[scale = .75]{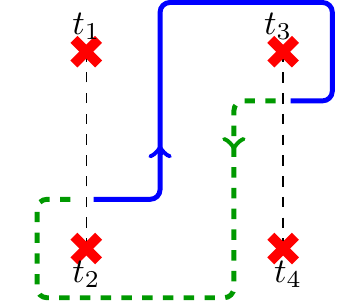}
    \subcaption{}
    \label{fig:Qd-cc-X-Z-dagger-5}
\end{subfigure}
~
\begin{subfigure}{.225\textwidth}
    \centering
    \includegraphics[scale = .75]{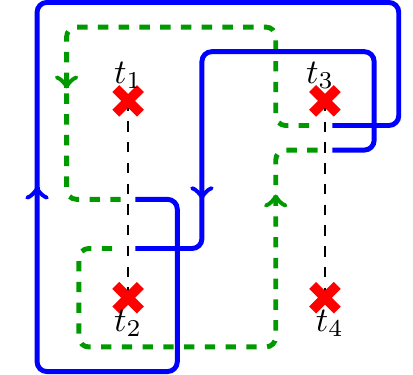}
    \subcaption{}
    \label{fig:Qd-cc-X-Z-dagger-6}
\end{subfigure}
~
\begin{subfigure}{.225\textwidth}
    \centering
    \includegraphics[scale = .75]{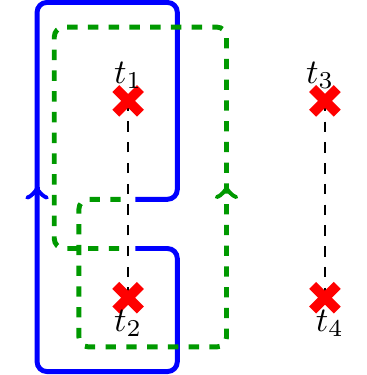}
    \subcaption{}
    \label{fig:Qd-cc-X-Z-dagger-7}
\end{subfigure}
~
\begin{subfigure}{.225\textwidth}
    \centering
    \includegraphics[scale = .75]{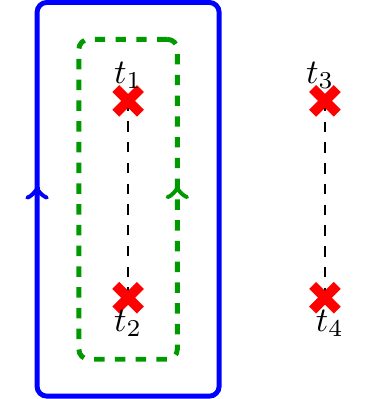}
    \subcaption{}
    \label{fig:Qd-cc-X-Z-dagger-8}
\end{subfigure}
~
\begin{subfigure}{.45\textwidth}
    \centering
    \includegraphics[scale = .75]{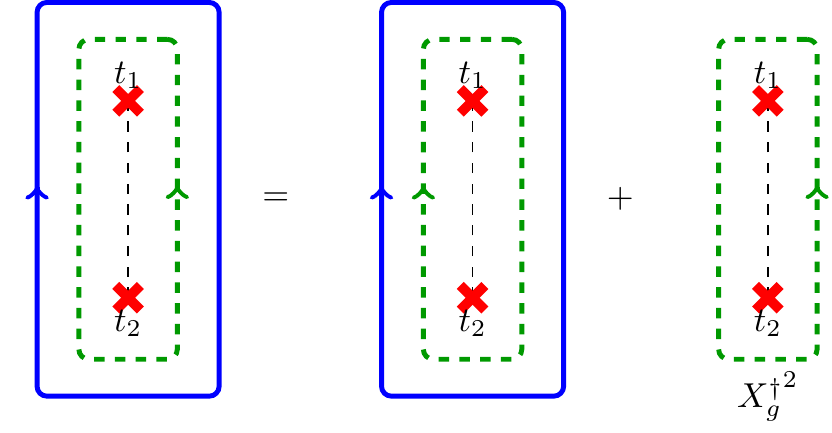}
    \subcaption{}
    \label{fig:Qd-cc-X-Z-dagger-9}
\end{subfigure}
\caption{Obtaining the string corresponding to operator $\Bar{X} \Bar{Z}^\dagger$ and its equivalence to the deformed string in Fig.~\ref{fig:Qd-cc-phase-gate-3}. (a) The operators $\Bar{X}$ and $\Bar{Z}^\dagger$ are shown. Note that the direction of $\Bar{Z}^\dagger$ is counterclockwise. (b) Combining the two strings shown in Fig.~\ref{fig:Qd-cc-X-Z-dagger-1}, we get the string as shown. (c) We add the stabilizer of twist $t_2$ to the string in Fig.~\ref{fig:Qd-cc-X-Z-dagger-2}. (d) The resulting string from Fig.~\ref{fig:Qd-cc-X-Z-dagger-3} double crosses itself. (e) The string in Fig.~\ref{fig:Qd-cc-X-Z-dagger-4} can be deformed further as shown. This string corresponds to $\Bar{X} \Bar{Z}^\dagger$. (f) The deformed string in Fig.~\ref{fig:Qd-cc-phase-gate-3} and the string in Fig.~\ref{fig:Qd-cc-X-Z-dagger-5} are combined. (g) Combination of strings in Fig.~\ref{fig:Qd-cc-X-Z-dagger-6} results in the string as shown. (h) By adding stabilizer of twist $t_2$, the string in Fig.~\ref{fig:Qd-cc-X-Z-dagger-7} is deformed as shown. (g) The string shown in Fig.~\ref{fig:Qd-cc-X-Z-dagger-8} is equivalent to a gauge operator.}
\label{fig:Qd-cc-X-Z-dagger}
\end{figure}

Before showing the equivalence between the operators shown in Fig.~\ref{fig:Qd-cc-phase-gate-3} and Fig.~\ref{fig:Qd-cc-phase-gate-4},  we consider the combination of operators $\Bar{X}$ and $\Bar{Z}^\dagger$.
The logical operators $\Bar{Z}^\dagger$ and $\Bar{X}$ are shown in Fig.~\ref{fig:Qd-cc-X-Z-dagger-1} and their combination is shown in Fig.~\ref{fig:Qd-cc-X-Z-dagger-2}.
To remove self-crossing, We add the stabilizer of twist $t_2$ to the string shown in Fig.~\ref{fig:Qd-cc-X-Z-dagger-2}, see Fig.~\ref{fig:Qd-cc-X-Z-dagger-3}.
The resulting string is shown in Fig.~\ref{fig:Qd-cc-X-Z-dagger-4} which can be deformed to the string shown in Fig.~\ref{fig:Qd-cc-X-Z-dagger-5}.
The string shown in Fig.~\ref{fig:Qd-cc-X-Z-dagger-5} corresponds to the operator $\Bar{X} \Bar{Z}^\dagger$.

We now show that the operator corresponding to the deformed string shown in Fig.~\ref{fig:Qd-cc-phase-gate-3} is equivalent up to gauge to the operator corresponding to string in Fig.~\ref{fig:Qd-cc-X-Z-dagger-5} 
Note that we add the conjugate of the operator shown in Fig.~\ref{fig:Qd-cc-X-Z-dagger-5}.
The combination of those strings is shown in Fig.~\ref{fig:Qd-cc-X-Z-dagger-6} can be simplified to the string shown in Fig.~\ref{fig:Qd-cc-X-Z-dagger-7}.
This string is further simplified by adding stabilizer of twist $t_2$ and obtain the string shown in Fig.~\ref{fig:Qd-cc-X-Z-dagger-8}.
The operator corresponding to the string in Fig.~\ref{fig:Qd-cc-X-Z-dagger-8} can be decomposed as shown in Fig.~\ref{fig:Qd-cc-X-Z-dagger-9}.
The blue and green string around twists is obtained by combining twist stabilizers and modified face stabilizers.
The green string with counterclockwise direction is the gauge operator, see Fig.~\ref{fig:Qd-cc-LO-charge-color-full}.
Therefore, the operator corresponding to the deformed string shown in Fig.~\ref{fig:Qd-cc-phase-gate-3} is equivalent up to gauge to the operator $\Bar{X} \Bar{Z}^\dagger$.

\section{An example}
In this section we present an example of qudit color code with two pairs of twists.
The smallest separation between any pair of twists is two.
The lattice with two pairs of twists is shown in Fig.~\ref{fig:Qd-cc-small-code}.
The additional faces (with two vertices) on the boundary are added to maintain trivalency.
Similar to normal faces, two stabilizers, one of $Z$ type and the other $X$ type are defined on them.

\begin{figure}[t]
    \centering
    \includegraphics[scale = .85]{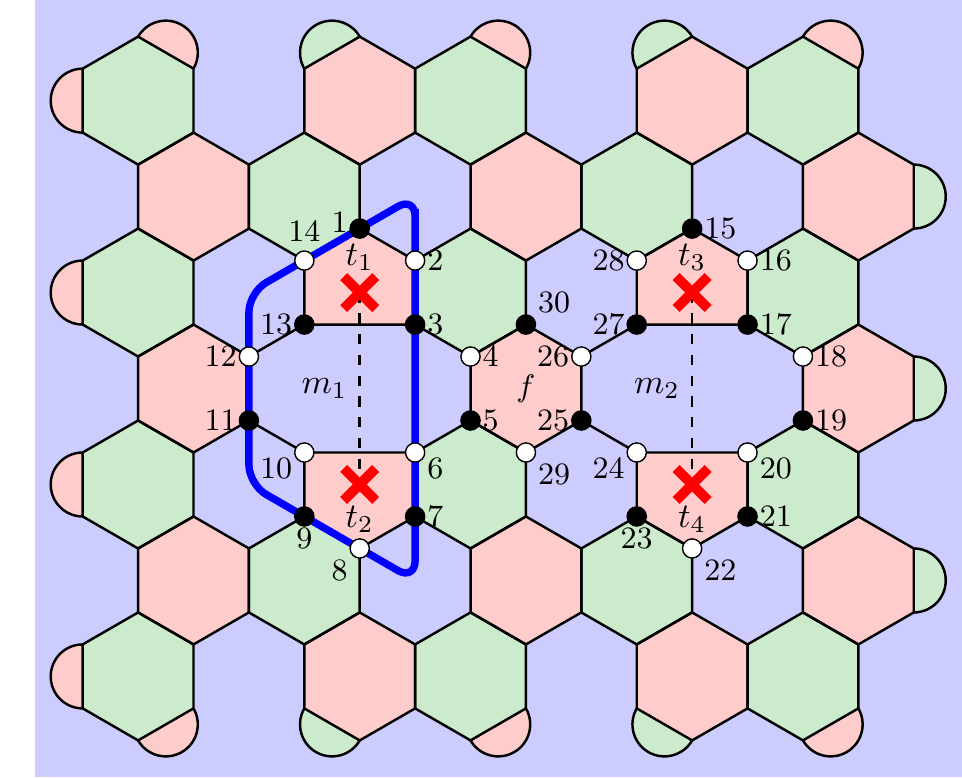}
    \caption{An example of qudit color code with twists.}
    \label{fig:Qd-cc-small-code}
\end{figure}

We now list explicitly the stabilizer generators defined on twist and modified faces.
The stabilizer generators defined on twist faces are
\begin{eqnarray}
B_{t_1} &=& (ZX)_{1} (ZX^\dagger)_{2} (ZX)_{3} (ZX)_{13} (ZX^\dagger)_{14},\\
B_{t_2} &=& (ZX^\dagger)_{6} (ZX)_{7} (ZX^\dagger)_{8} (ZX)_{9} (ZX^\dagger)_{10}, \\
B_{t_3} &=& (ZX^\dagger)_{20} (ZX)_{21} (ZX^\dagger)_{22} (ZX)_{23} (ZX^\dagger)_{24},\\
B_{t_4} &=& (ZX)_{15} (ZX^\dagger)_{16} (ZX)_{17} (ZX)_{27} (ZX^\dagger)_{28}.
\label{eqn:example-twist-stab}
\end{eqnarray}
Stabilizer generators defined on modified faces are
\begin{eqnarray}
B_{m_1,1} &=& X_{3} X_{4}^\dagger X_{5} X_{6}^\dagger Z_{10} Z_{11} Z_{12} Z_{13},\\
B_{m_1,2} &=& Z_{3} Z_{4} Z_{5} Z_{6} X_{10}^\dagger X_{11} X_{12}^\dagger X_{13}, \\
B_{m_2,1} &=& X_{17} X_{18}^\dagger X_{19} X_{20}^\dagger Z_{24} Z_{25} Z_{26} Z_{27},\\
B_{m_2,2} &=& Z_{17} Z_{18} Z_{19} Z_{20} X_{24}^\dagger X_{25} X_{26}^\dagger X_{27}.
\label{eqn:example-mod-stab}
\end{eqnarray}
On the normal face indicated $f$ in Fig.~\ref{fig:Qd-cc-small-code}, stabilizer generators are defined as below
\begin{eqnarray}
B_f^X &=& X_4^\dagger X_5 X_{25} X_{26}^\dagger X_{29}^\dagger X_{30}, \\
B_f^Z &=& Z_4 Z_5 Z_{25} Z_{26} Z_{29} Z_{30}.
\end{eqnarray}
The stabilizer generators defined on the additional red and green faces along the boundary are of the form $Z_u Z_v$ and $X_u X_v^\dagger$ where $u \in \mathsf{V}_e$ and $v \in \mathsf{V}_o$.

Note that in this example, we get two encoded logical qudits but we treat one of them as gauge qudit.
So in effect the lattice in Fig.~\ref{fig:Qd-cc-small-code} defines a  $[[118,1]]_q$ code.
The smallest nontrivial loop encircles twists created together see Fig.~\ref{fig:Qd-cc-small-code} where this operator is shown as blue string.
The operator corresponding to this string has weight $10$.


%

\end{document}